\documentclass{scrartcl}
\KOMAoptions{paper=letter}

\usepackage{changepage}
\usepackage[T1]{fontenc}

\usepackage{url}
\usepackage[hidelinks]{hyperref}
\usepackage[utf8]{inputenc}
\usepackage[small]{caption}
\usepackage{graphicx}
\usepackage{amsmath}
\usepackage{amsthm}
\usepackage{booktabs}
\usepackage{algorithm}
\usepackage{algorithmic}
\usepackage[round]{natbib}

\usepackage{changepage}
\usepackage{thm-restate}

\usepackage{tikz}
\usetikzlibrary{patterns}
\usepackage{cleveref}
\usepackage{array}
\usepackage{graphicx}
\usepackage{xcolor}
\usepackage{enumitem}
\usepackage{tabularx}
\usepackage{xspace}
\usepackage{amssymb}
\usepackage{amsmath}
\usepackage{amsthm}
\usepackage{bigfoot}

\newcommand{\sd}{\ensuremath{\mathit{SD}}\xspace}
\newcommand{\rd}{\ensuremath{\mathit{RD}}\xspace}
\newcommand{\cond}{\ensuremath{\mathit{COND}}\xspace}

\theoremstyle{definition}
\newtheorem{definition}{Definition}
\newtheorem{example}{Example}
\theoremstyle{plain}

\newtheoremstyle{bfnote}
{}{}%
{}{}%
{\bfseries}{.}%
{ }%
{\thmname{#1}\thmnumber{ #2}\thmnote{\textnormal{ (#3)}}}
\theoremstyle{bfnote}
\newtheorem{remark}{Remark}

\title{Strategyproof Randomized Social Choice for Restricted Sets of Utility Functions\thanks{A preliminary version of this paper appeared in Proceedings of the 30th International Joint Conference on Artificial Intelligence \citep{Lede21c}.}}

\author{Patrick Lederer\thanks{Email address: p.lederer@unsw.edu.au}\\
UNSW Sydney}

\newcolumntype{L}[1]{>{\raggedright\let\newline\\\arraybackslash\hspace{0pt}}m{#1}}
\newcolumntype{C}[1]{>{\centering\let\newline\\\arraybackslash\hspace{0pt}}m{#1}}
\newcolumntype{R}[1]{>{\raggedleft\let\newline\\\arraybackslash\hspace{0pt}}m{#1}}

\newenvironment{profile}{\medmuskip=0mu\relax
	\thickmuskip=1mu\relax \centering 
	\tabular}{\endtabular\medskip}

\newcommand{\profilewidth}{\columnwidth}

\newcommand{\ep}{\emph{ex post}\xspace}

\newcommand{\omni}{\ensuremath{\mathit{OMNI^*}}\xspace}

\begin{document}

\maketitle

\begin{abstract}
	Social decision schemes (SDSs) map the voters' preferences over multiple alternatives to a probability distribution over these alternatives. In a seminal result, \citet{Gibb77a} has characterized the set of SDSs that are strategyproof with respect to all utility functions and his result implies that all such SDSs are either unfair to the voters or alternatives, or they require a significant amount of randomization. To circumvent this negative result, we propose the notion of $U$-strategyproofness which postulates that only voters with a utility function in a predefined set $U$ cannot manipulate. We then analyze the tradeoff between $U$-strategyproofness and various decisiveness notions that restrict the amount of randomization of SDSs. In particular, we show that if the utility functions in the set $U$ value the best alternative much more than other alternatives, there are $U$-strategyproof SDSs that choose an alternative with probability $1$ whenever all but $k$ voters rank it first. On the negative side, we demonstrate that $U$-strategyproofness is incompatible with Condorcet-consistency if the set $U$ satisfies minimal symmetry conditions. Finally, we show that no \ep efficient and $U$-strategyproof SDS can be significantly more decisive than the uniform random dictatorship if the voters are close to indifferent between their two favorite alternatives.
\end{abstract}

\textbf{Keywords:} Randomized Social Choice, Social Decision Schemes, Strategyproofness, Decisiveness\medskip

\textbf{JEL Classification Code:} D71

	\section{Introduction}

    One of the central challenges in collective decision-making is strategic manipulation: voters may lie about their true preferences to influence the outcome in their favor. Such strategic manipulations are undesirable for a number of reasons. Firstly, when voters lie about their true preferences, it becomes difficult to identify socially desirable alternatives, so the overall quality of the decision-making process may deteriorate. Secondly, the desirable properties of voting rules are in doubt when voters act strategically because these properties are typically shown under the assumption that voters reveal their preferences truthfully. Finally, manipulable voting rules incentivize voters to learn the preferences of other voters and to identify optimal manipulations. Since this requires resources that are not evenly distributed in the population, voting becomes unfair.
    
    For these reasons, it is desirable to use strategyproof voting rules, which are immune to strategic manipulations by the voters. However, in a seminal result, \citet{Gibb73a} and \citet{Satt75a} have independently shown that every reasonable deterministic voting rule fails strategyproofness. Specifically, these authors have proven that every strategyproof voting rule is either dictatorial (i.e., it always chooses the favorite alternative of a fixed voter) or imposing (i.e., some alternatives can never be chosen). Clearly, this means that no strategyproof voting rule is acceptable for practical purposes.
    
    A natural escape route from this impossibility theorem is to allow for voting rules that may use chance to determine the winner. Such randomized voting rules, typically called \emph{social decision schemes (SDSs)}, assign to every alternative a probability based on the preferences of the voters and the final winner of the election is determined by chance according to these probabilities. For instance, the uniform random dictatorship is a well-known SDS which picks a voter uniformly at random and implements his favorite alternative as the winner of the election. Moreover, SDSs are typically considered strategyproof if no voter can increase his expected utility by misreporting his preferences for every utility function that is consistent with his true preference relation. In other words, this strategyproofness notion, which we call \sd-strategyproofness, ensures that voters cannot manipulate the outcome in their favor regardless of their exact utility~function. 
    
    Unfortunately, \citet{Gibb77a} and \citet{Barb79a} have shown that \sd-strat\-egy\-proof\-ness does not allow for particularly desirable SDSs. In more detail, these authors have characterized the set of \sd-strategyproof SDSs and it follows from their results that all such SDSs are either unfair to voters or alternatives or indecisive as they require a significant amount of randomization. 
    For instance, Gibbard's work implies that the uniform random dictatorship is the only \sd-strategyproof SDS that is anonymous (i.e., it treats all voters equally) and unanimous (i.e., an alternative is guaranteed to be selected if all voters report it as their favorite alternative). Put differently, this result characterizes the uniform random dictatorship as the most decisive SDS that is fair and \sd-strategyproof. However, even this SDS requires randomization as soon as two voters disagree on their favorite alternatives and is thus indecisive.\footnote{We refer here also to the work of \citet[][]{BBS20a}, who have experimentally verified that the uniform random dictatorship tends to use a significant amount of randomization.} Moreover, a result by \citet{Beno02a} implies that {every} \sd-strategyproof SDS has a chance to select an alternative that is bottom-ranked by all but one voters. This lack of decisiveness disqualifies \sd-strategyproof SDSs from many applications, where such a high degree of randomization is unacceptable.

    \paragraph{Contribution.} In an attempt to circumvent this impossibility of fair, strategyproof, and decisive SDSs, we will introduce and analyze a weakening of \sd-strategyproofness called $U$-strategyproofness. The idea of this axiom is to require strategyproofness only for a predefined set of utility functions $U$ instead of all possible utility functions. More formally, an SDS is said to be $U$-strategyproof for a set of utility functions $U$ if no voter with utility function $u\in U$ can increase his expected utility by lying about his true preferences. This means that $U$-strategyproofness for the set of all utility functions is equivalent to $\sd$-strategyproofness and $U$-strategyproofness becomes weaker when the set $U$ gets smaller. Thus, this strategyproofness notion allows for a much more fine-grained analysis than \sd-strategyproofness: instead of simply dismissing an SDS as \sd-manipulable, $U$-strategyproofness enables us to investigate for which utility functions it is strategyproof. Similarly, $U$-strategyproofness allows us to pinpoint the source of impossibility theorems by analyzing the necessary utility functions for such results. Furthermore, we believe that $U$-strategyproofness is also relevant in practice as it seems plausible for many situations that not all utility functions need to be considered. For instance, voters are unlikely to assign similar utilities to very different alternatives, so we may omit such utility functions when reasoning about their strategic behavior.

    Based on $U$-strategyproofness, we investigate when strategyproofness is compatible with three decisiveness axioms that restrict the amount of randomization that SDSs can use. In more detail, our first decisiveness condition is a new axiom called $k$-unanimity, which postulates that an alternative should be chosen with probability $1$ if it is top-ranked by all but $k$ voters. Or, in other words, this axiom forbids randomization if there is an alternative that is top-ranked by a wast majority of the voters. Moreover, we will also study the compatibility of $U$-strategyproofness with the more established conditions of Condorcet-consistency (which requires that an alternative that beats all other alternatives in pairwise majority comparisons should be chosen with probability~$1$) and \emph{ex post} efficiency (which requires that Pareto-dominated alternatives should be chosen with probability $0$). 
    Based on these axioms, we show the following results. 

    \begin{enumerate}[label=(\arabic*), leftmargin=*]
        \item We first suggest two variants of the uniform random dictatorship, called $\mathit{RD}^k$ and \omni, which satisfy $k$-unanimity and $U$-strategyproofness for large sets of utility functions $U$. Roughly, $\mathit{RD}^k$ chooses an alternative with probability $1$ if it is top-ranked by at least $n-k$ voters and agrees with the uniform random dictatorship if no such alternative exists. On the other hand, \omni chooses an alternative with probability $1$ if it is top-ranked by a majority of the voters and otherwise returns the uniform lottery over the top-ranked alternatives. We show that these SDSs are strategyproof for voters who favor their favorite alternative much more than their second best alternative. To make this more formal, let $u(x)$ denote the utility a voter assigns to his $x$-th best alternative. Then, $\mathit{RD}^k$ is $U$-strategyproof for the set $U$ containing all utility functions $u$ with $u(1)-u(2)\geq k(u(2)-u(m))$ and \omni is $U$-strategyproof for the set $U$ containing all utility functions $u$ with $u(1)-u(2)\geq \sum_{j=3}^m u(2)-u(j)$ (\Cref{thma}). Moreover, for the class of rank-based SDSs, we show that our SDSs are close to optimally solving the tradeoff between $k$-unanimity and $U$-strategyproofness (\Cref{thmb}).
        \item We next analyze the compatibility of Condorcet-consistency and $U$-strategy\-proof\-ness. Unfortunately, it turns out that, if there are $m\geq 4$ alternatives, there is no Condorcet-consistent SDS that is $U$-strategyproof for any non-empty set $U$ satisfying minimal symmetry constraints (\Cref{thmc}). In particular, it suffices if $U$ contains one utility function for each preference relation, thus demonstrating a far-reaching impossibility. On the other hand, when $m=3$, we show that the Condorcet rule, which picks the Condorcet winner with probability $1$ whenever such an alternative exists and otherwise randomizes uniformly over all alternatives, is $U$-strategyproof for the set $U$ containing all utility functions $u$ with $u(1)-u(2)=u(2)-u(3)$. Moreover, we prove that, except for profiles with majority ties, this rule is characterized by Condorcet-consistency and $U$-strategyproofness for the given set $U$ and $m=3$ alternatives (\Cref{thmd}).
        \item Finally, we show that no \ep efficient and $U$-strategyproof SDS can be significantly more decisive than the uniform random dictatorship when the set $U$ contains utility functions that are close to indifferent between their two favorite alternatives. More precisely, we show that if $U$ contains utility functions $u$ such that $u(1)-u(2)\leq \frac{\epsilon}{2}(u(2)-u(3))$, then no $U$-strategyproof and \ep efficient SDS can guarantee more than $\frac{\ell}{n}+\epsilon$ probability to an alternative that is top-ranked by $\ell$ voters (\Cref{thm:expost1}). We note that this result can be seen as a generalization of the work by \citet{Beno02a}. Moreover, in combination with \Cref{thma,thmb}, \Cref{thm:expost1} has a simple interpretation: the design of strategyproof and decisive SDSs is only possible if voters themselves are decisive about their favorite alternatives. 
    \end{enumerate}

	\paragraph{Related Work.} Inspired by the Gibbard-Satterthwaite theorem \citep{Gibb73a,Satt75a} and the follow-up works by \citet{Gibb77a} and \citet{Barb79a} for randomized voting rules, a large body of literature has investigated escape routes from and variants of these impossibility results. We thus refer the reader to \citet{Tayl05a} and \citet{Barb10a} for a general overview of strategyproofness in social choice theory. 

    More specifically, our work belongs to a growing body of literature that investigates social decision schemes with respect to strategyproofness notions other than \sd-strat\-egy\-proofness \citep[e.g.,][]{Hoan17a,ABBB15a,BBEG16a,BLS22c,BrLe24a}. 
    Similar to our work, these papers typically observe that \sd-strategyproofness is incompatible with various other properties and they aim to circumvent these impossibility results by considering weaker strategyproofness notions. The main difference between our work and these existing results is that $U$-strategyproofness is much more versatile than previously considered strategyproofness notions: whereas all existing strategyproofness notions are based on a particular method of comparing lotteries, $U$-strategyproofness allows us to flexibly choose the considered set of utility functions and thus enables a more fine-grained analysis. Moreover, we note that weakenings of \sd-strategyproofness have also been considered in random assignment \citep[e.g.,][]{BoMo02a,Balb16a,Cho18a,ChYu20a}, and we believe that $U$-strategyproofness could also be of interest in this domain. 

    While $U$-strategyproofness has, to the best of our knowledge, not been studied before, similar concepts have been suggested. For instance, \citet{Sen11a} defined strategyproofness based on utility functions that put an emphasis on the best three alternatives, and \citet{MeSe21a} focus on utility functions where the utility exponentially decreases with the position of an alternative in the preference relation. Moreover, in set-valued social choice (where the outcome of an election is a non-empty set of alternatives instead of a lottery) preferences over sets of alternatives are often derived from utility functions. For instance, \citet{DuSc00a} and \citet{Beno02a} employ this approach to motivate their strategyproofness notions. 
    
    Furthermore, $U$-strategyproofness is also related to the study of cardinal SDSs, where voters report cardinal utilities instead of ordinal preference relations \citep[e.g.,][]{Hyll80a,DPS07a,Nand12a,EMMS20a}. In more detail, when restricting the domain of feasible utility functions to $U$, $U$-strategy\-proof SDSs induce strategyproof cardinal SDSs by replacing each voter's utility function with the induced ordinal preference relation. Thus, the study of $U$-strategyproofness can also be seen as the study of strategyproof cardinal SDSs for restricted domains of utility functions. In this sense, our paper is also related to a significant line of work that analyzes strategyproof SDSs on restricted domains of preferences \citep[e.g.,][]{EPS02a,BMS05a,ELP17a,ChZe18a,ChZe21a}. However, these authors restrict the set of feasible preference relations, whereas $U$-strategyproofness restricts the domain of feasible utility functions. 

    Finally, the tradeoff between strategyproofness and decisiveness has been considered before \citep[e.g.,][]{Barb77b,Beno02a,BBL21b,BLR23a}. These papers show that, under varying assumptions, strategyproofness and decisiveness are incompatible. By contrast, we precisely quantify this tradeoff by analyzing the set of utility functions for which decisive and strategyproof SDSs exist. Moreover, we note that the paper by \citet{BLR23a} is dual to ours as these authors investigate how well \sd-strategyproof SDSs can approximate Condorcet-consistency and \emph{ex post} efficiency. That is, \citeauthor{BLR23a} weaken decisiveness axioms while requiring full \sd-strategyproofness, and we weaken \sd-strategyproofness while requiring full decisiveness.

	\section{Preliminaries}\label{sec:preliminaries}
	
	Let $N=\{1,\dots, n\}$ denote a set of $n$ voters and $A=\{x_1,
    \dots,x_m\}$ a set of $m$ alternatives. We assume that every voter $i\in N$ reports a \emph{preference relation}~$\succ_i$, which is a complete, transitive, and anti-symmetric binary relation on $A$. We compactly represent preference relations as comma-separated lists and use the $*$ symbol as placeholder for omitted alternatives. For instance, $x_1,*,x_2$ means that $x_1$ is the most preferred alternative, $x_2$ the least preferred one, and that the order of the remaining alternatives is not specified. The set of all preference relations is given by $\mathcal{R}$. We summarize the preference relations of all voters by a \emph{preference profile}, which is formally a function that assigns each voter to a preference relation. The set of all preference profiles is therefore $\mathcal{R}^N$. We typically write preference profiles by indicating the sets of voters that report a specific preference relation directly before the preference relation.

	The study object of this paper are {social decision schemes}, which are voting rules that may use chance to determine the winner of an election. To make this more formal, we define \emph{lotteries} as probability distributions over the alternatives, i.e., a lottery $p$ is a function from $A$ to $[0,1]$ such that $\sum_{x\in A} p(x) =1$. Moreover, we denote by $\Delta(A)$ the set of all lotteries over $A$. Formally, a \emph{social decision scheme (SDS)} $f$ is a function that maps every preference profile $R\in\mathcal{R^N}$ to a lottery $p\in \Delta(A)$, i.e., the signature of an SDS is $f: \mathcal{R}^N\rightarrow \Delta(A)$. We denote by $f(R,x)$ the probability assigned to $x$ by $f$ in the profile $R$, and we interpret this term as the probability that $x$ will be chosen as the final winner in the profile $R$. Furthermore, we extend this notation to sets of alternatives $X\subseteq A$ by defining $f(R,X)=\sum_{x\in X} f(R,x)$.

\subsection{Fairness Axioms}

In this paper, we will focus on designing SDSs that treat all voters and alternatives fairly. In particular, all SDSs that we suggest will satisfy two basic fairness conditions called anonymity and neutrality, which formalize that voting rules should not discriminate against particular voters or alternatives. Moreover, we will also consider the class of rank-based SDSs. We note that most of our negative results do not require any of these fairness axioms.

    \paragraph{Anonymity.} Anonymity is a mild fairness (or symmetry) condition that postulates that all voters are treated equally. Formally, an SDS $f$ is \emph{anonymous} if $f(R)=f(\pi(R))$ for all preference profiles $R$ and permutations $\pi:N\rightarrow N$. Here, $R'=\pi(R)$ denotes the profile given by ${\succ'_{\pi(i)}}={\succ_{i}}$ for all $i\in N$.

    \paragraph{Neutrality.} Similar to anonymity, neutrality requires that all alternatives are treated fairly. We again formalize this concept with the help of permutations: an SDS $f$ is \emph{neutral} if $f(\tau(R),\tau(x))=f(R,x)$ for all preference profiles $R$, alternatives $x$, and permutations $\tau:A\rightarrow A$. This time, we denote by $R'=\tau(R)$ the profile given by $\tau(x)\succ_i \tau(y)$ if and only if $x\succ_i y$ for all $x,y\in A$ and $i\in N$. 

    \paragraph{Rank-basedness.} Rank-basedness is a strengthening of anonymity that requires that the outcome of an SDS should only depend only the positions of the alternatives in the voters' preference but not on their exact order. To make this more formal, we define the \emph{rank} of an alternative~$x$ in a preference relation $\succ_i$ by $r(\succ_i,x)=|\{y \in A\colon y\succ_i x\}|$. That is, a voter's favorite alternative has rank $1$, his second favorite alternative has rank $2$, and so on. The \emph{rank vector} $r^*(R,x)$ of an alternative $x$ contains the rank of~$x$ with respect to every voter in increasing order, i.e., $r^*(R,x)_i\leq r^*(R,x)_{i+1}$ for all $i\in \{1,\dots, n-1\}$. Moreover, the \emph{rank matrix} $r^*(R)$ contains the rank vectors of all alternatives as rows. Finally, we call an SDS $f$ \emph{rank-based} if it only depends on the rank matrix, i.e., $f(R)=f(R')$ for all preference profiles $R$ and $R'$ with $r^*(R)=r^*(R')$. The class of rank-based SDSs has been first considered by \citet{Lasl96a} and contains many prominent rules such as positional scoring rules and anonymous tops-only SDSs.

\subsection{Decisiveness Axioms}

    We will now introduce our decisiveness axioms, namely $k$-unanimity, Condorcet-consis\-tency, and \emph{ex post} efficiency. The first two of these conditions formalize decisiveness by requiring that randomization should only be used if there is no clear winner in the preference profile. We believe such conditions to be crucial to ensure the acceptability of randomization in voting and that the outcome chosen by the SDSs faithfully follows the voters' preferences. Indeed, whenever an SDS randomizes over multiple alternatives, the choice of the final winner depends on chance instead of the voters' preferences, which seems particularly undesirable if there is a clear winner in a preference profile. On the other hand, \ep efficiency formalizes decisiveness in a dual sense by forbidding randomization over Pareto-dominated~alternatives. 

    \paragraph{$k$-Unanimity.} A common but very weak decisiveness condition is unanimity, which requires that an alternative is guaranteed to be chosen if it is unanimously top-ranked by all voters. More formally, an SDS $f$ is \emph{unanimous} if $f(R,x)=1$ for all preference profiles $R$ and alternatives $x$ such that all voters top-rank $x$ in $R$. While this condition is uncontroversial, it seems too weak for practical applications as elections are rarely unanimous. We thus generalize this notion by requiring that an alternative is chosen if all but $k$ voters report it as their favorite alternative. Specifically, an SDS $f$ is \emph{$k$-unanimous} if $f(R,x)=1$ whenever $n-k$ or more voters report $x$ as their favorite alternative. We note that unanimity is equivalent to $0$-unanimity and that $k$-unanimity is only satisfiable if $k<\frac{n}{2}$. Moreover, $k$-unanimity generalizes existing axioms: for instance, \citeauthor{Beno02a}'s (\citeyear{Beno02a}) near unanimity is equivalent to $1$-unanimity and the absolute winner property of \citet{BLS22c} corresponds to $\lfloor\frac{n-1}{2}\rfloor$-unanimity. Also, \citet{Amor09a} studies a concept called unequivocal majority, which is closely related to $k$-unanimity. 

    \paragraph{Condorcet-consistency.} Condorcet-consistency is a prominent decisiveness axiom that requires that the Condorcet winner is chosen with probability $1$ whenever it exists. To formalize this condition, let $n_{xy}(R)=|\{i\in N\colon x \succ_i y\}|-|\{i\in N\colon y \succ_i x\}|$ denote the \emph{majority margin} between two alternatives $x,y\in A$ in the preference profile~$R$. An alternative $x$ is the \emph{Condorcet winner} in a preference profile~$R$ if $n_{xy}(R)>0$ for all other alternatives $y\in A\setminus \{x\}$, i.e., if it beats every other alternative in a pairwise majority comparison. Moreover, an SDS $f$ is \emph{Condorcet-consistent} if $f(R,x)=1$ for all profiles $R$ and alternatives $x\in A$ such that $x$ is the Condorcet winner in $R$. We note that Condorcet-consistency implies $k$-unanimity for every $k<\frac{n}{2}$ as an alternative that is top-ranked by more than half of the voters is the Condorcet winner.

    \paragraph{Ex post efficiency.} The idea of \ep efficiency is that an alternative should have no chance of winning the election if there is a unanimously more preferred alternative. In more detail, we say an alternative $x$ \emph{Pareto-dominates} another alternative $y$ in a preference profile $R$ if $x\succ_i y$ for all voters $i\in N$. Conversely, an alternative $x$ is \emph{Pareto-optimal} in a profile $R$ if it is not Pareto-dominated by any other alternative. Finally, an SDS $f$ is \emph{ex post} efficient if it only randomizes over Pareto-optimal alternatives, i.e., $f(R,x)=0$ for all preference profiles $R$ and Pareto-dominated alternatives $x$. We observe that \ep efficiency implies unanimity.

\subsection{$U$-Strategyproofness} \label{subsec:USP}

    The central axiom in our analysis is strategyproofness, which postulates that voters should not be able to benefit by lying about their true preferences. Following the standard approach in the literature \citep[e.g.,][]{Gibb77a,Sen11a,BBEG16a}, we will formalize this axiom by assuming that every voter $i\in N$ is endowed with a \emph{(von Neumann-Morgenstern) utility function $u_i:A\rightarrow \mathbb{R}$} and compares lotteries via their expected utility. Put differently, a voter $i$ with utility function $u_i$ prefers lottery $p$ to lottery $q$ if $u_i(p)=\sum_{x\in A} p(x) u_i(x) \geq \sum_{x\in A} q(x) u_i(x)=u_i(q)$. We say that a utility function~$u_i$ is \emph{consistent} with a preference relation ${\succ_i}$ if it ordinally agrees with $\succ_i$, i.e., $u_i(x)> u_i(y)$ if and only if $x\succ_i y$ for all distinct $x,y\in A$. Moreover, we define by $\mathcal{U}^{\succ}=\{u\in \mathbb{R}^A\colon \text{$u$ is consistent with $\succ$}\}$ the set of all utility functions that are consistent with the preference relation $\succ$ and by $\mathcal{U}=\bigcup_{{\succ}\in\mathcal{R}} \mathcal{U}^{\succ}$ the set of all (injective) utility functions. Finally, given an integer $k\in\{1,\dots, m\}$, we will often write $u_i(k)$ for the utility of the alternative with the $k$-th highest utility. This allows us to conveniently define constraints on utility functions that are independent of particular preference relations. For instance, $U=\{u\in\mathcal{U}\colon u(1)\geq 2 u(2)\}$ is the set of utility functions that assign at least twice as much utility to the best alternative than to the second-best one.
	
	Even though we assume the existence of utility functions, voters only report ordinal preference relations. As a consequence, strategyproofness is commonly defined by quantifying over all utility functions that are consistent with a voters' preference relation, which results in the standard notion of \sd-strategyproofness.\footnote{\sd stands for stochastic dominance as \sd-strategyproofness can equivalently be defined via stochastic dominance \citep[e.g.,][]{Sen11a,BBEG16a}.\nopagebreak}

    \begin{definition}[\sd-strategyproofness]
        An SDS $f$ is \emph{\sd-strategyproof} if $u_i(f(R))\geq u_i(f(R'))$ for all voters $i\in N$, preference profiles $R$, $R'$, and utility functions ${u_i\in \mathcal U}$ such that $u_i$ is consistent with $\succ_i$ and ${\succ_j}={\succ_j'}$ for all $j\in N\setminus \{i\}$.
    \end{definition}

    As usual, we say an SDS is \emph{\sd-manipulable} if it is not \sd-strategyproof. We note that \sd-strategyproofness is predominant in the literature as it guarantees that voters cannot manipulate regardless of their exact utility functions \citep[e.g.,][]{Gibb77a,Barb79a,EPS02a,BLR23a}. However, as discussed in the introduction, this strategyproofness notion leads to rather negative results and it seems for many practical applications not necessary to prohibit manipulations with respect to all utility functions. Moreover, \sd-strategyproofness is somewhat crude as it does not allow us to pinpoint the types of voters for which an SDS is strategyproof or manipulable.
    Motivated by these observations, we introduce a new strategyproofness notion called $U$-strategyproofness, which weakens $\sd$-strategyproofness by only quantifying over a predefined subset of utility functions $U$ instead of the set of all utility functions $\mathcal{U}$. 

    \begin{definition}[$U$-strategyproofness]
        An SDS $f$ is \emph{${U}$-strategyproof} if $u_i(f(R))\geq u_i(f(R'))$ for all voters $i\in N$, preference profiles $R$, $R'$, and {utility functions ${u_i\in U}$} such that $u_i$ is consistent with $\succ_i$ and ${\succ_j}={\succ_j'}$ for all $j\in N\setminus \{i\}$.
    \end{definition}

    Analogous to \sd-strategyproofness, we say an SDS is \emph{$U$-manipulable} if it fails $U$-strategyproofness, which means that there is a profile $R$, a voter $i$, and a utility function $u_i\in U$ such that voter $i$ can improve his expected utility with respect to $u_i$ by lying about his true preferences in $R$. Moreover, we emphasize that $U$-strategyproofness only guarantees that voters with a utility function in $U$ cannot manipulate. Thus, $U$-strategyproofness becomes weaker when we consider a smaller set of utility functions $U$, to the point where it is trivial when $U=\emptyset$. On the other extreme, $\mathcal{U}$-strategyproofness, i.e., $U$-strategyproofness with respect to the set of all utility functions $\mathcal{U}$, is equivalent to $\sd$-strategyproofness. Hence, $U$-strategyproofness allows us to define a spectrum of strategyproofness notions that weaken \sd-strategyproofness.
    
    We will next discuss an example to illustrate the difference between $U$-strategyproofness and \sd-strategyproofness. 
	
	\begin{example}\label{ex}
		Consider the profiles $R^1$ and $R^2$ shown below and let $f$ denote an SDS such that $f(R^1,x)=\frac{1}{3}$ for $x\in \{a,b,c\}$ and $f(R^2,b)=1$. Moreover, consider the utility functions $u_1$, $u_2$, and $u_3$ with $u_1(a)=2$, $u_1(b)=1$, $u_1(c)=0$, $u_2(a)=3$, $u_2(b)=1$, $u_2(c)=0$, $u_3(a)=3$, $u_3(b)=2$, and $u_3(c)=0$. These utility functions are consistent with voter $1$'s preference relation in $R^1$ and we check whether this voter can benefit by deviating to $R^2$. A quick calculation shows that $u_1(f(R^1))=1=u_1(f(R^2))$, $u_2(f(R^1))=\frac{4}{3}> 1=u_2(f(R^2))$, and $u_3(f(R^1))=\frac{5}{3}<2=u_3(f(R^2))$. Hence, voter $1$ can increase his expected utility if his utility function is $u_3$ and $f$ is \sd-manipulable. By contrast, voter $1$ does not benefit from deviating to $R^2$ if his utility function is $u_1$ or $u_2$. Since the preferences of the other voters are not consistent with $u_1$, $u_2$, and $u_3$, it follows that $f$ is $\{u_1, u_2\}$-strategyproof on these two profiles.\smallskip
		
		\begin{profile}{C{0.07\profilewidth} C{0.3\profilewidth} C{0.3\profilewidth} C{0.3\profilewidth}}
			$R^1$: & $1$: $a, b, c$ & $2$: $b,c,a$ & $3$: $c,a,b$ \\
			$R^2$: & $1$: $b, a, c$ & $2$: $b,c,a$ & $3$: $c,a,b$
		\end{profile}
	\end{example}
	
	In our results, we will typically examine $U$-strategyproofness for \emph{symmetric} sets $U$, i.e., we assume that $u \in U$ implies that $u \circ \pi \in U$ for every permutation $\pi$ on $A$. This formalizes the natural condition that all preference relations should be treated equally by strategyproofness. Furthermore, this symmetry condition is rather weak since every neutral SDS is $U$-strategyproof for a symmetric set $U$ (See Claim (1) in \Cref{prop:properties}). Moreover, we will often restrict our attention to the case where $U$ is given by a single utility function $u$ and its renamings, i.e., $U=\{u\circ\pi\colon \pi\in \Pi\}$, where $\Pi$ denotes the set of all permutations on $A$. In this case, we write $u^\Pi$-strategyproofness instead of $U$-strategyproofness. We emphasize that $u^\Pi$-strategyproofness associates every preference relation with exactly one utility function, i.e., it is the weakest non-trivial form of $U$-strategyproofness for a symmetric set $U$. 
    
    We conclude this section by discussing several helpful properties of $U$-strategyproofness. We defer the proofs of the following statements to the appendix.

    \begin{restatable}{proposition}{properties}\label{prop:properties}
        Consider a non-empty set of utility functions $U$ and suppose that $f$ and~$g$ are two $U$-strategyproof SDSs. The following claims are true:
        \begin{enumerate}[label=(\arabic*), leftmargin=*]
            \item If $f$ is neutral, it is $U'$-strategyproof for a symmetric set $U'$ with $U\subseteq U'$.
            \item $f$ is $U'$-strategyproof for the set $U'=\bigcup_{\succ\in\mathcal{R}} \text{conv}(U\cap \mathcal{U}^{\succ})$, where $\text{conv}(X)$ denotes the convex hull of a given set $X$. 
            \item The SDS $h$ given by $h(R)=\lambda f(R)+(1-\lambda) g(R)$ is $U$-strategyproof for all $\lambda\in [0,1]$.
            \item It holds that $u(f(R))\geq u(f(R'))$ for all preference profiles $R$ and $R'$, groups of voters $S\subseteq N$, and utility functions $u\in U$ such that ${\succ_i}={\succ_j}$ for all $i,j\in S$, $u$ is consistent with $\succ_i$ for all $i\in S$, and ${\succ_j}={\succ_j'}$ for all $j\in N\setminus S$.  
        \end{enumerate}
    \end{restatable}

    Less formally, Claim (1) states that, for neutral SDSs, it is without loss of generality to focus on symmetric sets of utility functions. Of course, if $f$ is not symmetric, this is not the case as a rule may only be manipulable for specific preference relations. Claim (2) shows that for every SDS $f$ and preference relation $\succ$, $f$ is $U$-strategyproof for a convex subset $U$ of $\mathcal{U}^{\succ}$. In particular, this implies that for every SDS there is a unique inclusion-maximal set $U$ for which it is $U$-strategyproof. Claim (3) shows that the set of SDSs that is $U$-strategyproof for a given set of utility functions $U$ is convex and mirrors an analogous insight for \sd-strategyproofness. Finally, Claim (4) can be seen as a mild version of group-strategyproofness: if an SDS is $U$-strategyproofness, it cannot be manipulated by groups of voters who have the same true preferences, even if the voters in a group cooperate. This insight will be useful in our proofs as it allows us to simultaneously change the preference relations of multiple voters.

	\section{Results}\label{sec:Results}
	
	We are now ready to present our results. Specifically, we will analyze the compatibility of $U$-strategyproofness with $k$-unanimity, Condorcet-consistency, and \emph{ex post} efficiency in \Cref{subsec:kunanimity,subsec:Condorcet,subsec:expost}, respectively.
    
	\subsection{$k$-Unanimity}\label{subsec:kunanimity}
    
    As our first question, we will investigate the tradeoff between $U$-strategyproofness and $k$-unanimity. To this end, we first recall that the only \sd-strategyproof SDS that satisfies anonymity and unanimity is the uniform random dictatorship (henceforth called $\rd$). This SDS chooses a voter uniformly at random and implements his favorite alternative as the winner of the election. Hence, the probability that an alternative $x$ is the winner in a profile $R$ is $\mathit{RD}(R,x)=\frac{|\{i\in N\colon r(\succ_i,x)=1\}|}{n}$. However, $\rd$ fails $k$-unanimity for every $k>0$ and, more generally, \citet{Beno02a} has shown that no \sd-strategyproof SDS is $k$-unanimous for $k>0$.
	
	To circumvent this impossibility, we will next define a variant of $\rd$ that satisfies $k$-unanimity for an arbitrary $k\in \{0, \dots, \lfloor \frac{n-1}{2}\rfloor\}$ and $U$-strategyproofness for a large set of utility functions $U$. Consider for this the following family of SDSs, which we call $k$-random dictatorships (abbreviated by $\rd^k$): if at least $n-k$ voters report $x$ as their favorite alternative, $\rd^k$ assigns probability $1$ to $x$; otherwise, it returns the outcome of \rd. As we show in \Cref{thma}, $\rd^k$ satisfies $U$-strategyproofness for the set $U=\{u\in\mathcal U\colon u(1)-u(2)\geq k (u(2) - u(m))\}$, i.e., if voters have a strong preference for the first alternative, $\rd^k$ is strategyproof. Unfortunately, the definition of $U$ depends on $k$, i.e., for large values of $k$, there must be a very large gap between $u(1)$ and $u(2)$. Another variant of $\rd$, which we refer to as $\omni$, solves this problem. This SDS assigns probability $1$ to an alternative $x$ if more than half of the voters report $x$ as their favorite alternative, and otherwise randomizes uniformly over all alternatives that are top-ranked by at least one voter. This SDS is $U$-strategyproof for $U=\{u\in\mathcal U\colon u(1)-u(2)\geq \sum_{i=3}^m u(2)-u(i)\}$. While $\omni$ satisfies $\lfloor \frac{n-1}{2}\rfloor$-unanimity for all numbers of voters and alternatives, the condition on $U$ is demanding unless there are only few alternatives. 
	
	\begin{restatable}{theorem}{thma}\label{thma} The following claims are true:
    \begin{enumerate}[label=(\arabic*), leftmargin=*]
        \item For all $k\in \{1,\dots,\lfloor\frac{n-1}{2}\rfloor\}$, $\rd^k$ satisfies $U$-strategyproofness for the set $U=\{u\in\mathcal U\colon u(1)-u(2)\geq k (u(2) - u(m))\}$ and violates $\{u\}$-strategyproofness for every utility function $u\not\in U$.
        \item $\omni$ satisfies $U$-strategyproofness for the set $U=\{u\in \mathcal U\colon u(1)-u(2)\geq \sum_{i=3}^m u(2) - u(i)\}$ and violates $\{u\}$-strategyproofness for every utility function $u\not\in U$ if $n\geq m$. 
    \end{enumerate}
	\end{restatable}
    \begin{proof}
        We prove both claims of the theorem separately.\medskip

        \textbf{Claim (1):}
		First, we show that $\rd^k$ is $U$-strategyproof for the set $U=\{u\in\mathcal U\colon u(1)-u(2)\geq k(u(2)-u(m))\}$ for all $k\in \{1,\dots,\lfloor\frac{n-1}{2}\rfloor\}$. To this end, we fix a voter $i\in N$, a utility function $u\in U$, and preference profiles $R$ and $R'$ such that ${\succ_j}={\succ_j'}$ for all $j\in N\setminus \{i\}$ and $u$ is consistent with $\succ_i$. We will show that  $u(\rd^k(R))\geq u(\rd^k(R'))$. If neither $R$ nor $R'$ contain $n-k$ voters who agree on a most preferred alternative, $\rd^k$ is equal to $\rd$ for both profiles. Because $\rd$ is \sd-strategyproof, $\rd^k$ is $U$-strategyproof in this case. Moreover, voter $i$ can also not manipulate if at least $n-k$ voters agree on a most preferred alternative $x$ in $R$. In more detail, if voter $i$ top-ranks $x$ in $R$, he obtains his maximal utility since $\rd^k(R,x)=1$ and thus $u(\rd^k(R))\geq u(\rd^k(R'))$. On the other hand, if voter $i$ does not top-rank $x$ in $R$, at least $n-k$ voters top-rank $x$ in $R'$. This means that $\rd^k(R)=\rd^k(R')$ and therefore also $u(\rd^k(R))=u(\rd^k(R'))$.
    		
		The only remaining case is that $n-k-1$ voters top-rank an alternative $x$ in $R$, voter $i$ prefers another alternative $y$ the most, and the remaining $k$ voters top-rank alternatives in $A\setminus \{x\}$. In this case, voter $i$ might try to manipulate by reporting $x$ as his favorite alternative. In particular, this means for $R'$ that $\rd^k(R',x)=1$. On the other hand, it holds by the definition of $R$ that $\rd^k(R,x)=\frac{n-k-1}{n}$ and $\rd(R,y)\geq \frac{1}{n}$. Now, let $z$ denote voter $i$'s least preferred alternative in $R$. By our analysis so far, it holds that $u(\rd^k(R))\geq \frac{n-k-1}{n} u(x)+\frac{1}{n} u(y)+\frac{k}{n} u(z)$ and that $u(\rd^k(R'))=u(x)$. Next,
        $\frac{n-k-1}{n} u(x)+\frac{1}{n} u(y)+\frac{k}{n} u(z)\geq u(x)$ is true if and only if $u(y)-u(x) \geq k(u(x)-u(z))$. Finally, since $u(y)=u(1)$, $u(x)\leq u(2)$, and $u(z)=u(m)$, it follows from the definition of $U$ that  
        \[u(y)-u(x)\geq u(1)-u(2)\geq k(u(2)-u(m))\geq k(u(x)-u(m)).\]
        This shows that voter $i$ cannot increase his expected utility in this case either, so we conclude that $\rd^k$ is $U$-strategyproof for $U=\{u\in\mathcal U\colon u(1)-u(2)\geq k(u(2)-u(m))\}$.

        Finally, to show that $\rd^k$ fails $\{u\}$-strategyproofness for every utility function $u$ with $u(1)-u(2)<k(u(2)-u(m))$, we fix such a utility function $u$ and let $x,y,z$ denote the alternatives with $u(y)=u(1)$, $u(x)=u(2)$, and $u(z)=u(m)$. Moreover, consider the profile $R$ where $n-k-1$ voters top-rank $x$, $k$ voters top-rank $z$, and a single voter~$i$ reports a preference relation $\succ_i$ that is consistent with $u$. In particular, this means that $y$ is voter $i$'s best alternative, $x$ is his second best alternative, and $z$ is his worst alternative. For this profile, $\rd^k$ assigns probability $\frac{n-k-1}{n}$ to $x$, $\frac{1}{n}$ to $y$, and $\frac{k}{n}$ to~$z$. Hence, voter $i$'s expected utility is $u(\rd^k(R))=\frac{n-k-1}{n} u(2)+\frac{1}{n}u(1)+\frac{k}{n}u(m)$. By contrast, if voter $i$ top-ranks $x$, his expected utility is $u(2)$, and it can be verified that $u(2)>\frac{n-k-1}{n} u(2)+\frac{1}{n}u(1)+\frac{k}{n}u(m)$ if $u(1)-u(2)<k(u(2)-u(m))$. This proves that $\rd^k$ fails $\{u\}$-strategyproofness for every utility function $u$ with $u(1)-u(2)<k(u(2)-u(m))$.\medskip

        \textbf{Claim (2):} Next, we show that $\omni$ is $U$-strategyproof for the set $U=\{u\in \mathcal{U}\colon u(1)-u(2)\geq \sum_{i= 3}^m (u(2) - u(i))\}$. To this end, we again fix a voter $i$, a utility function $u\in U$, and preference profiles $R$ and $R'$ such that ${\succ_j}={\succ_j'}$ for all $j\in N\setminus \{i\}$ and $u$ is consistent with $\succ_i$. We will show that $u(\omni(R))\geq u(\omni(R'))$. For this, we proceed with a case distinction on the lotteries chosen for $R$ and $R'$. First, assume that $\omni(R,x)=1$ for some alternative $x\in A$, which means that more than half of the voters report $x$ as their best alternative. Now, if $i$ top-ranks $x$, he obtains his maximal utility and he cannot manipulate. On the other hand, if $i$ does not top-rank $x$, then a majority of the voters top-ranks $x$ also in $R'$, which implies that $\omni(R')=\omni(R)$ and therefore $u(\omni(R'))=u(\omni(R))$.
		
		Next, consider the case that $\omni$ returns for both $R$ and $R'$ the uniform lottery over the top-ranked alternatives of the respective profiles. Let $S=\{z\in A \colon \omni(R,z)>0\}$ and $T=\{z\in A \colon \omni(R',z)>0\}$ denote the sets of alternatives with positive winning chance in $R$ and $R'$, respectively. Moreover, let $x$ be the most preferred alternative of voter $i$ in $R$, and let $y$ be his most preferred alternative in $R'$. If voter $i$ is the only voter in $R$ that top-ranks $x$ in $R$, then $T=(S\setminus \{x\})\cup \{y\}$. If $y\not\in S$, this means that $\omni(R,z)=\omni(R',z)$ for all $z\in A\setminus \{x,y\}$, $\omni(R',y)=\omni(R',x)$, and $\omni(R',x)=\omni(R,y)=0$. Since $u(x)>u(y)$ as voter $i$ prefers $x$ to $y$, it follows that $u(\omni(R))\geq u(\omni(R'))$. On the other hand, if $y\in S$, then it holds tht $\omni(R',z)=\frac{1}{|S|-1}$ for all $z\in S\setminus \{x\}$ as only $|S|-1$ alternatives are top-ranked. Put differently, this means that we redistributed the probability of $x$ to the alternatives in $S\setminus \{x\}$. Since voter $i$ top-ranks $x$ in $R$, we have that $u(x)>u(z)$ for all $z\in S\setminus \{x\}$, so it follows again that $u(\omni(R))\geq u(\omni(R'))$. 
        		
		As the second subcase, suppose that another voter top-ranks voter $i$'s best alternative~$x$. Then, it holds that $T=S\cup \{y\}$. If $y$ was already top-ranked in $R$, it further follows that $y\in S$, so the outcome does not change. On the other hand, if $y\not\in S$, then $\omni(R',z)=\frac{1}{|S|+1}$ for all $z\in S\cup \{y\}$ and $\omni(R,z)=\frac{1}{|S|}$ for all $z\in S$. In this case, we compute that 
        \begin{align*}u(\omni(R))-u(\omni(R'))&=\sum_{z\in S} \left(\frac{1}{|S|}-\frac{1}{|S|+1}\right)u(z) - \frac{1}{|S|+1}u(y)\\
        &=\frac{1}{|S|(|S|+1)} \sum_{z\in S} (u(z) - u(y)).
        \end{align*}
        Moreover, since $u\in U$, it holds that $u(1)-u(2)\geq \sum_{j=3}^m u(2)-u(j)$, which equivalently means $(u(1)-u(2))+\sum_{j=3}^{m} u(j)-u(2)\geq 0$. Since voter $i$'s favorite alternative $x$ is in $S$ and $u(y)\leq u(2)$, we derive that 
        \begin{align*}
            \sum_{z\in S} (u(z) - u(y))\geq \sum_{z\in S} (u(z) - u(2))\geq u(1)-u(2)+\sum_{j=3}^m u(j)-u(2)\geq 0.
        \end{align*}
        Combined with our last equation, this means means $u(\omni(R))\geq u(\omni(R'))$, thus proving that voter $i$ cannot manipulate in this case. 
		
		As the last case, assume that the set $S=\{z\in A\colon \omni(R,z)>0\}$ contains at least two alternatives and that $\omni(R',y)=1$ for some alternative $y\in A$. 
        This is only possible if voter $i$ reports $y$ as his favorite alternative in $R'$ but not in $R$. 
        Similar to the last case, we hence compute that $u(\omni(R))-u(\omni(R'))=\frac{1}{|S|} \sum_{z\in S} u(z)-u(y)$. 
        Moreover, since $u\in U$, it follows that $\sum_{z\in S} u(z)-u(y)\geq 0$ because voter $i$'s favorite alternative $x$ is in $S$ and $u(y)\leq u(2)$. 
        Hence, voter $i$ cannot manipulate in this case either, which proves that $\omni$ is $U$-strategyproof for the set $U=\{u\in \mathcal U\colon u(1)-u(2)\geq \sum_{i=3}^m u(2)-u(i)\}$.  
		
		Finally, we show that $\omni$ violates $\{u\}$-strategyproofness for every utility function~$u$ with $u(1)-u(2)< \sum_{i=3}^m u(2)-u(i)$ if $n\geq m$. For this, we fix such a utility function $u$ and note that our assumption equivalently means that $u(2)>\frac{\sum_{j=1}^m u(j)}{m}$. Moreover, suppose that the alternatives $x_1,\dots, x_m$ are ordered such that $u(x_1)>\dots>u(x_m)$. 
        Now, if $m=3$, consider the profile $R$ where $x_2$ is top-ranked by $\lfloor\frac{n}{2}\rfloor$ voters, $\lceil\frac{n}{2}\rceil-1$ voters top-rank $x_3$ and a single voter $i$ reports the preference relation $x_1,x_2,x_3$. For this profile, $\omni$ assigns a probability to $\frac{1}{3}$ to each alternative, so $u(\omni(R))=\frac{1}{m}\sum_{j=1}^m u(j)$. On the other hand, if voter $i$ top-ranks $x_2$, it holds that this alternative is chosen with probability $1$. This means that voter $i$ can manipulate since $u(2)>\frac{\sum_{j=1}^m u(j)}{m}$. Next, if $m\geq 4$, consider a profile $R$ where voter $i$ reports the preference relation $x_1,x_2,\dots,x_m$, every alternative $x_j\in A\setminus \{x_2\}$ is top-ranked by a voter other than $i$, and no alternative is top-ranked by more than half of the voters. For this profile, $\omni$ randomizes uniformly over the alternatives $A\setminus \{x_2\}$, so the utility of voter $i$ is $\frac{1}{m-1}(u(1)+\sum_{j=3}^m u(j))$. On the other hand, if voter $i$ top-ranks $x_2$, $\omni$ randomizes uniformly over $A$, resulting in an expected utility of $\frac{1}{m}\sum_{j=1}^m u(j)$. Finally, it can be checked that $\frac{1}{m}\sum_{j=1}^m u(j)>\frac{1}{m-1}(u(1)+\sum_{j=3}^m u(j))$ if $u(1)-u(2)< \sum_{i=3}^m u(2)-u(i)$, thus demonstrating that voter $i$ can indeed manipulate. 
	\end{proof}
	
	While it is positive that $k$-unanimity and $U$-strategy\-proofness can be simultaneously satisfied, the sets $U$ for which $\rd^k$ and $\omni$ are $U$-strategyproof become rather restrictive for large values of $k$ and $m$. This raises the question of whether there are SDSs that satisfy $U$-strategyproofness and $k$-unanimity for larger sets of utility functions~$U$. As we show next, at least for rank-based SDSs, our results are already close to optimal as no such SDS satisfies $k$-unanimity and $U$-strategyproofness for a significantly larger set of utility functions than $\rd^k$ or \omni. 
	
	\begin{restatable}{theorem}{thmb}\label{thmb}
		No rank-based SDS satisfies $u^\Pi$-strategy\-proofness and $k$-unanimity for $k\in \{1,\dots,\lfloor\frac{n-1}{2}\rfloor\}$ if $m\geq 3$, $n\geq 3$, and $u(1)-u(2)<\sum_{i=\max(3, m-k+1)}^m u(2)-u(i)$. 
	\end{restatable}
	\begin{proof}
		Fix some value $k\in \{1,\dots, \lfloor\frac{n-1}{2}\rfloor\}$ and let $k^*=\min(k,m-2)$, which means that $\sum_{i=m-k^*+1}^m u(2)-u(i)=\sum_{i=\max(3,m-k+1)}^m u(2)-u(i)$. We assume for contradiction that there is a rank-based SDS $f$ for $m\geq 3$ alternatives and $n\geq 3$ voters that satisfies $k$-unanimity and $u^\Pi$-strategyproofness for some utility function $u$ with $u(1)-u(2)<\sum_{i=m-k^*+1}^m u(2)-u(i)$. For deriving a contradiction, we proceed in two steps: first, we discuss a construction that allows us to weaken an alternative $x$ that is currently assigned probability $1$ from first place to second place without affecting the outcome if sufficiently many voters top-rank $x$. Secondly, we repeatedly use this construction to derive a profile $R^*$ in which $x$ gets probability $1$ even though only $k^*$ voters report it as their best choice. Moreover, we can ensure that the remaining $n-k^*\geq n-k$ voters top-rank another alternative $y$, so $k$-unanimity is violated for $R^*$.\medskip
		
		\textbf{Step 1:} Let $S=\{x_0, \dots, x_{k^*}\}$ denote a set of $k^*+1$ alternatives and let $\hat{x}_i=x_{i\text{ mod } (k^*+1)}$ to simplify notation. Moreover, let $x^*$ denote an alternative not in $S$ and note that such an alternative exists as $k^*\leq m-2$. Our goal is to find profiles $R^0,\dots, R^{k^*}$ such that (i) $r^*(R^i)=r^*(R^j)$ for all $i,j\in \{0,\dots, k^*\}$, and (ii) in every profile $R^i$, there is a voter $j^*$ with preference $\hat x_i, x^*, *, \hat x_{i+1}, \hat x_{i+2},\dots, \hat x_{i+k^*}$. Given these profiles, we show that $f(R^i,x^*)=1$ if $f(\hat{R}^i, X^*)=1$ for all $i\in \{0,\dots, {k^*}\}$, where $\hat R^i$ denotes the profile derived from $R^i$ by letting voter $j^*$ swap his favorite alternative $\hat  {x}_i$ with $x^*$. For the sake of simplicity, we focus subsequently on the case that there are $n=2k^*+1$ voters. If there are more voters, we can pick a suitable subset of $2k^*+1$ voters and apply our construction to these voters while keeping the preferences of the other voters constant. 
		
		We now define the profiles $R^0,\dots, R^{k^*}$. In the profile $R^i$, the voters $j\in \{1,\dots, k^*+1\}$ with $j\neq i$ have the preference relations $x^*,\hat x_{j},*, \hat x_{j+1}, \dots,  \hat x_{j+k^*}$. Moreover, voter $i$ has the preference $\hat x_{i}, x^*, *, \hat x_{i+1}, \dots, \hat x_{i+k^*}$, i.e., the construction of voter $i$'s preference relation differs from the preference relations of the other voters $j$ with $j\leq k^*+1$ only in that his best two alternatives are swapped. Note that the alternatives $\{x_0,\dots, x_{k^*}\}$ form a ''cylic'' subprofile for these $k^*+1$ voters. Next, the voters $j\in \{k^*+2,\dots, 2k^*+1\}$ with $j\neq k^*+1+i$ have the preference relation $\hat x_{j}, \hat x_0, *, \hat x_1, \dots, \hat x_{j-1}, \hat x_{j+1}, \dots, \hat x_{k^*},x^*$. Finally, the voter $j=k^*+1+i$ has the same preference relation except that he swaps $\hat x_{j}=x_i$ with $\hat x_0=x_0$, i.e., ${\succsim^i_{k^*+1+i}}=\hat x_0, \hat x_{j}, *, \hat x_1, \dots, \hat x_{j-1}, \hat x_{j+1}, \dots, \hat x_{j+k^*},x^*$. We note that in $R^0$, alternative $x_0=\hat x_0$ is the second ranked by all voters $j$ with $j\geq k^*+2$ because there is no voter with $j\geq k^*+2$ and $j=k^*+1+0$. Moreover we assume that all voters have the same preferences on the alternatives in $Y=A\setminus \{x^*,x_0, \dots, x_{k^*}\}$ (these alternatives were abbreviated by the $*$ symbol in all previous preference relations). 
		
		Note that all profiles $R^i$ have the same rank matrix because $r^*(R^i, x)=r^*(R^0,x)$ for all $x\in A$ and $i\in \{1, \dots, k^*\}$. For the alternatives in $Y$, this claim holds since every voter ranks these alternatives at the same positions in all of our profiles. For the alternatives in $A\setminus Y$, this follows because the profile $R^0$ differs from every other profile $R^i$ only in the preference relations of voters $i$, $k^*+1$, and $k^*+i+1$. Moreover, the preferences of these voters only differ in swaps between $x^*$, $x_0=\hat x_{k+1}$, and $x_i=\hat x_i$. In more detail, $\succ_i^i$ is derived from $\succ_i^0$ by reinforcing $x_i$ against $x^*$, $\succ_{k^*+1}^i$ is derived from $\succ_{k^*+1}^0$ by reinforcing $x^*$ against $x_0$, and $\succ_{k^*+i+1}^i$ is derived from $\succ_{k^*+i+1}^0$ by reinforcing $x_0$ against $x_i$. Because all these swaps happen between first and second ranked alternatives, the rank vectors of the alternatives in $A\setminus Y$ are equal in the profiles $R^0$ and $R^i$. This proves that $r^*(R^0)=r^*(R^i)$ for all profiles $R^i$, and rank-baseness consequently implies that $f(R^i)=f(R^0)=f(R^j)$ for all $i,j\in \{0, \dots, k^*\}$. 
		
		Finally, it remains to show that $f(R^i,x^*)=1$ for all $i\in \{0, \dots, k^*\}$. We suppose for this that $f(\hat R^i,x^*)=1$ for all $i\in \{0, \dots, k^*\}$, where $\hat R^i $ denotes the profiles derived from $R^i$ by reinforcing $x^*$ against $\hat x_i$ in the preference of voter $i$ (or $k^*+1$ if $i=0$). By our assumption, this voter can ensure that his expected utility is $u(2)$ if he deviates to $\hat R^i$. Hence, $u^\Pi$-strategyproofness entails that the expected utility of voter $i$ in $R^i$ is at least $u(2)$, which means that the following inequality holds. Note that we assume here that $u$ is permuted to be consistent with voter $i$'s preference relation.
		\begin{align*}
		f(R^i, \hat x_{i})u(1) + f(R^i, x^*)u(2) 
		+ \sum_{y\in Y} f(R^i,y)u(y) 
		 + \sum_{j=1}^{k^*} f(R^i, \hat x_{i+j}) u(\hat x_{i+j})\geq u(2)
		\end{align*}
		
		We reformulate this inequality to highlight the similarity to our condition on the utility function $u$. 	
		\begin{align*}
		f(R^i, \hat x_{i}) (u(1) - u(2))& \geq \sum_{y\in Y} f(R^i,y) (u(2) - u(y)) + \sum_{j=1}^{k^*} f(R^i, \hat x_{i+j}) (u(2) - u(\hat x_{i+j})) 
		\end{align*}
		
		Furthermore, we derive an analogous inequality for every profile $R^j$ with $j\in \{0,\dots, k^*\}$. Since $f(R^i)=f(R^j)$ for all $i,j\in \{0, \dots, k^*\}$, we can substitute $f(R^j,x)$ with $f(R^i,x)$ in these inequalities. Moreover, for every $i\in \{0, \dots, {k^*}\}$, alternative $\hat x_i$ is top-ranked by the manipulator in $R^i$, and for every $r\in \{m-k^*+1, \dots, m\}$, there is a single profile $R^j$ such that the manipulator ranks $\hat x_i$ at position $r$. Hence, we derive the following inequality by summing up all constraints on $f(R^i)$.
		\begin{align*}
		\sum_{j=0}^{k^*} f(R^i, x_{j})  (u(1) - u(2)) \geq \,\,&
		(k^*+1)\sum_{y\in Y} f(R^i,y) (u(2) - u(y)) \\
		&+ \sum_{j=0}^{k^*} f(R^i, x_{j}) \sum_{\ell=m-k^*+1}^m (u(2) - u(\ell))
		\end{align*}

        {\sloppy
        Next, we note that $u(1)-u(2)<\sum_{j=m-k^*+1}^m (u(2)-u(j))$ by assumption, so 
        \[\sum_{j=0}^{k^*} f(R^i, x_{j})  (u(1) - u(2))<\sum_{j=0}^{k^*} f(R^i, x_{j}) \sum_{l=m-k^*+1}^m (u(2) - u(l)).\] 
        As $u(2)-u(y)>0$ for all $y\in Y$, it follows that the above inequality can only be true if $f(R^i,x)=0$ for all $x\in A\setminus \{x^*\}$. In turn, this implies that $f(R^i,x^*)=1$, as desired.}\medskip
		
		\textbf{Step 2:} We next use Step 1 to derive a profile $R$ in which $x^*$ is top-ranked by only $k^*$ voters but it is chosen with probability $1$. For this, we start at the profile $\bar{R}^0$ in which the first $n-k^*$ voters prefer alternative $x^*$ the most and the remaining voters prefer $x^*$ uniquely the least. It follows from $k$-unanimity that $f(\bar{R}^0,x^*)=1$ as $k^*\leq k$. Moreover, $u^\Pi$-strategyproofness entails that all voters can reorder the alternatives in $A\setminus \{x^*\}$ arbitrarily without affecting the outcome. In more detail, if $x^*$ does not obtain probability $1$ after a voter who top-ranks $x^*$ reorders the remaining alternatives, this voter can manipulate by undoing his deviation. On the other hand, if a voter who prefers $x^*$ the least can reduce the probability of $x^*$ by changing his preferences, he can increase his expected utility for every utility function. 
        
        Hence, we can pick a subset $I$ of voters who prefer $x^*$ the most with $|I|=k^*+1$ and the $k^*$ voters who prefer $x^*$ the least and assign them the preferences in the profile $\hat{R}^i$ for every $i\in \{0, \dots, k^*\}$ without affecting the outcome. Consequently, we can use the last step to derive a profile $\bar{R}^1$ such that $f(\bar{R}^1, x^*)=1$, $n-k^*-1$ voters prefer $x^*$ the most, one voter reports $x_0, x^*, *, x_1, \dots, x_{k^*}$, and the remaining voters prefer $x^*$ the least. Moreover, it is easy to see that we can repeat this step as long as at least $k^*+1$ voters top-rank $x^*$. By repeatedly applying this construction, we derive a profile $\bar{R}$ such that $f(\bar R, x^*)=1$, $k^*$ voters prefer $x^*$ the most, $n-2k^*$ voters report $x_0, x^*, *, x_1, \dots, x_{k^*}$, and $k^*$ voters report $x^*$ as their least preferred outcome. 
		
		Finally, we recall that the voters who prefer $x^*$ the least can reorder the alternatives in $A\setminus \{x^*\}$ arbitrarily without affecting the outcome if $x^*$ is chosen with probability $1$. We thus let these voters top-rank $x_0$ and $u^\Pi$-strategyproofness requires for the resulting profile $R^*$ that $f(R^*,x^*)=1$. However, in $R^*$, $n-k^*\geq n-k$ voters top-rank $x_0$, so $k$-unanimity requires that $f(R^*,x_0)=1$. This is the desired contradiction which shows that no rank-based SDS is both $k$-unanimous and $u^{\pi}$-strategyproof for a utility function $u$ with $u(1)-u(2)<\sum_{i=m-k^*+1}^m u(2)-u(i)$.
	\end{proof}

    To better understand the results of this section, we illustrate our bounds in \Cref{fig1}. For this figure, we assume that there are $5$ alternatives and a large number of voters $n\geq 11$, and we fix all utilities but $u(1)$. Hence, we can compute the values of $u(1)$ for all SDSs of \Cref{thma} such that the considered SDS is $u^\Pi$-strategyproof. The figure shows that for $\rd^k$, the required value of $u(1)$ increases in $k$ and the bound of $\omni$ is independent of $k$. Moreover, the required values of $u(1)$ are quite large compared to $u(2)$ for all SDSs but \rd. However, the red area shows the values of $u(1)$ for which \Cref{thmb} applies and hence, these large values are indeed required. The white area indicates a small gap between the positive results in \Cref{thma} and the impossibility theorem in \Cref{thmb} for $1<k<\lfloor\frac{n-1}{2}\rfloor$.

       	\begin{figure}\centering
		\begin{tikzpicture}
		\fill[red!30] (1,1) -- (1,2) -- (2,2) -- (2,8/3) -- (3,8/3) -- (3,9/3) -- (5,9/3) -- (5,1); 
		\fill[black!30] (0,0) -- (0,1) -- (5,1) -- (5,0); 
		\fill[pattern=north west lines, pattern color=blue!70] (0,1) -- (1,1) -- (1,4) -- (0,4); 
		\fill[pattern=north east lines, pattern color=green!70] (0,2) -- (2,2) -- (2,4) -- (0,4);
		\fill[pattern=horizontal lines, pattern color=magenta] (0,3) -- (3,3) -- (3,4) -- (0,4); 
		\fill[pattern=crosshatch dots, pattern color=orange] (0,3) -- (5,3) -- (5,4) -- (0,4); 
		\draw[line width = 0.15mm, ->] (0, 0) -- (5.2, 0) node[right] {$k$};
		\draw[line width = 0.15mm, ->] (0, 0) -- (0, 4.2) node[above] {$u(1)$};
		\draw[line width = 0.1mm] (1, 0) -- (1, -0.2) node[below] {1};
		\draw[line width = 0.1mm] (2, 0) -- (2, -0.2) node[below] {2};
		\draw[line width = 0.1mm] (3, 0) -- (3, -0.2) node[below] {3};
		\draw[line width = 0.1mm] (4, 0) -- (4, -0.2) node[below] {4};
		\draw[line width = 0.1mm] (5, 0) -- (5, -0.2) node[below] {5};
		
		\draw[line width = 0.1mm] (0,1) -- (-0.2,1) node[left] {3};
		\draw[line width = 0.1mm] (0,2) -- (-0.2,2) node[left] {6};
		\draw[line width = 0.1mm] (0,3) -- (-0.2,3) node[left] {9};
		\draw[line width = 0.1mm] (0,4) -- (-0.2,4) node[left] {12};
		
		\node[align=center] at (3.5,2) {no rank-based \\SDSs};
		\node at (0.5, 1.5) {$\rd$};
		\node at (1.5, 2.5) {$\rd^1$};
		\node at (2.5, 3.5) {$\rd^2$};
		\node at (4, 3.5) {$\omni$};
		\draw[line width = 0.2mm] (0, 1) -- (5.2,1) node[right] {$u(2)$};
		\end{tikzpicture}\vspace{-0.3cm}

		\caption{Illustration of \Cref{thma,thmb}. We assume that there are $5$ alternatives and consider a utility function $u$ with $u(2)=3$, $u(3)=2$, $u(4)=1$, and $u(5)=0$. The figure shows for which values of $u(1)$ the SDSs $\rd$ (blue area), $\rd^1$ (green area), $\rd^2$ (magenta area), and $\omni$ (orange area) are $u^\Pi$-strategyproof on the vertical axis. The horizontal axis illustrates the values of $k$ for which these SDSs are $k$-unanimous. The red area displays the impossibility of \Cref{thmb} and the gray area marks the infeasible values of $u(1)$ with $u(1)\leq u(2)$.}
		\label{fig1}
	\end{figure}
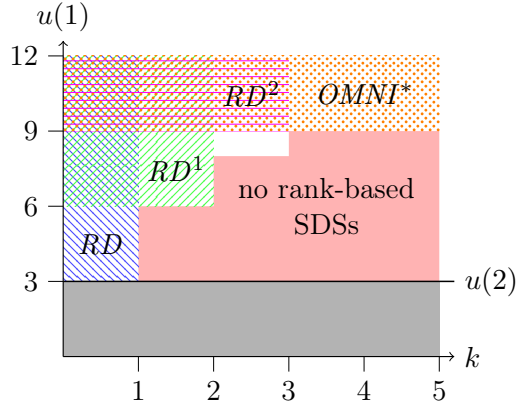

    	\begin{remark}\label{rem1}
		 Most of the bounds of \Cref{thmb} are tight: if $m=2$, \omni and $\rd^k$ are even \sd-strategyproof, and if $n=2$, $k$-unanimity is not well-defined for $k>0$. Furthermore, the condition on the utility functions is almost tight: $\rd^1$ shows that the bound is tight for $1$-unanimity, and $\omni$ shows that the bound is tight if $k\geq m-2$. Finally, $\rd^{k}$ shows that no constraint of the type $u(1)-u(2)\leq \sum_{i=m-k+1}^m u(2)-u(i) + \epsilon$ with $\epsilon>0$ can result in an impossibility because we can always find a utility function $u$ such that $\sum_{i=m-k+1}^m u(2)-u(i) + \epsilon \geq u(1)-u(2)\geq k(u(2)-u(m))$ by making the difference between $u(i)$ and $u(m)$ for $i\geq 3$ sufficiently small. Nevertheless, it remains open to find rank-based SDSs that satisfy $U$-strategyproofness and $k$-unanimity for $U=\{u\in\mathcal U\colon u(1)-u(2)=\sum_{i=m-k+1}^m u(2)-u(i)\}$ and $2\leq k\leq m-3$. 
	\end{remark}

    \begin{remark}
        It can be shown for every $k\in \{1,\dots,\lfloor\frac{n}{2}\rfloor\}$ and set $U\subseteq \mathcal{U}$ that if there is a $k$-unanimous and $U$-strategyproof SDS $f$ for $m$ alternatives and $n$ voters, there is also an SDS $f'$ that satisfies these conditions for $m$ alternatives $n'>n$ voters. In particular, $f'$ can be constructed by choosing a set of $n$ voters uniformly at random and applying $f$ to the preferences of these voters. This insight can be used to show that \Cref{thmb} is tight for $2$-unanimity. Specifically, it can be verified that the following SDS $f^1$ is $U$-strategyproof for the set $U=\{u\in\mathcal{U}\colon u(1)-u(2)\geq 2u(2)-u(m-1)-u(m)\}$ if $n=5$: if at least four alternatives are top-ranked, $f^1$ agrees with \rd, and it agrees with \omni otherwise. By our previous observation, we can generalize this SDS to an arbitrary number of voters $n>5$ while guaranteeing $U$-strategyproofness and $2$-unanimity. 
    \end{remark}

    \begin{remark}
        When omitting rank-basedness, one can define SDSs that are $k$-unanimous and $U$-strategyproof for larger sets of utility functions $U$ than permitted by \Cref{thmb}. For instance, when $m=3$ and $n=4$, we have shown with the help of a computer that the following SDS $f^2$ is $U$-strategyproof for the set $U=\{u\in\mathcal{U}\colon 2(u(2)-u(3))\geq u(1)-u(2)\geq \frac{1}{2}(u(2)-u(3))\}$. If three or four voters top-rank an alternative $x$, then $f^2(R,x)=1$. If two alternatives $x,y$ are top-ranked by two voters each, then $f^2(R,x)=f^2(R,y)=\frac{1}{2}$. Finally, if an alternative $x$ is top-ranked by two voters, and the other two alternatives are top-ranked once, we need a further case distinction depending on the voters who do not top-rank $x$: if both of these voters place $x$ second, then $f^2(R,x)=1$. If both place $x$ last, then $f^2(R,x)=\frac{1}{2}$ and $f^2(R,y)=f^2(R,z)=\frac{1}{4}$. Finally, if one of them places $x$ second and the other one last, then $f^2(R,x)=\frac{4}{7}$, the alternative $y$ that is preferred to $x$ by two voters obtains probability $f^2(R,y)=\frac{2}{7} $, and the last alternative $z$ has probability $f^2(R,z)=\frac{1}{7}$. While our computer-aided approach can also be used for larger values of $m$ and $n$, we see little value in such technical SDSs with unclear closed-form representation.
    \end{remark}
	
	\subsection{Condorcet-consistency}\label{subsec:Condorcet}
	Motivated by the existence of SDSs that are $k$-unanimous and $U$-strategyproof for non-empty sets of utility functions $U$, we next turn to the question of whether stronger fairness notions are still compatible with $U$-strategyproofness. Unfortunately, we find a negative answer to this question when considering Condorcet-consistency. In particular, we will show that it is impossible to design Condorcet-consistent SDSs that are $u^\Pi$-strategyproof for any utility function $u\in \mathcal{U}$. We note that this result significantly strengthens the known impossibility of \sd-strategyproof and Condorcet-consistent SDSs \citep[e.g.,][]{BLR23a} as our theorem only relies on a single canonical utility function.
	
	\begin{restatable}{theorem}{thmc}\label{thmc}
		Every Condorcet-consistent SDS fails $u^\Pi$-strategyproofness for all utility functions $u\in \mathcal U$ if $m\geq 4$, $n\geq 5$ and $n\neq6$, $n\neq8$. 
	\end{restatable}
	\begin{proof}
		Assume for contradiction that there is a Condorcet-consistent SDS $f$ for $m\geq 4$ alternatives and $n\geq 5$ voters ($n\neq 6$, $n\neq 8$) that satisfies $u^\Pi$-strategyproofness for some utility function $u\in\mathcal{U}$. The proof works by a case distinction: first, we show that there is no $u^\Pi$-strategyproof SDS that is Condorcet-consistent if $m\geq 4$, $n=3$, and $u(1)-u(2)<u(2)-u(m)$. Next, we show that there is no $u^\Pi$-strategyproof SDS that satisfies Condorcet-consistency if $m\geq 4$, $n=5$, and $u(1)-u(m-1)>u(m-1)-u(m)$. These two cases are exhaustive with respect to the utility functions, i.e., every utility function on at least $4$ alternatives satisfies $u(1)-u(2)<u(2)-u(m)$ or $u(1)-u(m-1)>u(m-1)-u(m)$. Specifically, since $u(2)>u(m-1)$, it holds that $u(1)-u(m-1)>(m-1)-u(m)$ if $u(1)-u(2)\geq u(2)-u(m)$. Finally, since we prove both cases for a fixed number of voters $n$, we generalize the impossibility from a fixed number of voters to larger numbers of voters in the last step.\bigskip
		
		\textbf{Case 1: $u(1)-u(2)<u(2)-u(m)$}
		
		As the first case, we assume that $f$ is defined for $n=3$ voters and satisfies $u^\Pi$-strategyproofness for a utility function $u$ with $u(1)-u(2)<u(2)-u(m)$. Moreover, we let $X=\{x,y,z\}$ denote three alternatives and consider the following preference profiles.\smallskip
		
		\begin{profile}{C{0.07\profilewidth} C{0.29\profilewidth} C{0.29\profilewidth} C{0.29\profilewidth}}
			$R^{1}$: & $1$: $x,y,*,z$ & $2$: $y,z,*,x$ & $3$: $z,x,*,y$\\
			$R^{2}$: & $1$: $y,x,*,z$ & $2$: $y,z,*,x$ & $3$: $z,x,*,y$\\ 		  
			$R^{3}$: & $1$: $x,y,*,z$ & $2$: $z,y,*,x$ & $3$: $z,x,*,y$\\
			$R^{4}$: & $1$: $x,y,*,z$ & $2$: $y,z,*,x$ & $3$: $x,z,*,y$\\	  
		\end{profile}\smallskip

        First, we note that $y$ is the Condorcet winner in $R^2$, $z$ in $R^3$, and $x$ in $R^4$. Consequently, Condorcet-consistency requires that $f(R^2,y)=f(R^3,z)=f(R^4,x)=1$. Moreover, $R^1$ differs from $R^2$ only in the preference relation of the first voter, from $R^3$ in the preference relation of the second voter, and from $R^4$ in the preference relation of the third voter. Hence, we can use $u^\Pi$-strategyproofness to derive constraints on $f(R^1)$. In particular, we derive the following inequality from $u^\Pi$-strategyproofness between $R^1$ and $R^2$. 	
		\begin{align*}
		f(R^1,x)u(1)+f(R^1,y)u(2)+f(R^1,z)u(m)+\sum\limits_{w\in A\setminus X} f(R^1, w)u(w)\geq u(2)
		\end{align*}
		
         Moreover, we derive symmetric conditions from $u^\Pi$-strategyproofness between $R^1$ and $R^3$, and between $R^1$ and $R^4$. By reformulating these inequalities, we deduce that
		\begin{align*}
		f(R^1,x)(u(1) - u(2)) & \geq f(R^1,z)(u(2) - u(m))
		 +\sum_{w\in A\setminus X} f(R^1, w) (u(2) - u(w)),\\
		f(R^1,y)(u(1) - u(2)) & \geq f(R^1,x)(u(2) - u(m))
		 +\sum_{w\in A\setminus X} f(R^1, w) (u(2) - u(w)),\\
		f(R^1,z)(u(1) - u(2)) & \geq f(R^1,y)(u(2) - u(m))
		 +\sum_{w\in A\setminus X} f(R^1, w) (u(2) - u(w)).
		\end{align*}
		
		By summing up these inequalities, we derive the following equation. 
		\begin{align*}
		\sum_{w\in X} f(R^1,w)(u(1) - u(2)) \geq &
		 \sum_{w\in X} f(R^1,w)(u(2) - u(m)) +
		  3\!\!\!\sum_{w\in A\setminus X}\!\!\! f(R^1, w) (u(2) - u(w))
		\end{align*}
		
		Because $u(1) - u(2)<u(2)-u(m)$, it follows that $\sum_{w\in X} f(R^1,w)(u(1) - u(2))< \sum_{w\in X} f(R^1,w)(u(2) - u(m))$ if $\sum_{w\in X} f(R^1,w)>0$. Moreover, it holds that $u(2)>u(w)$ for every $w\in A\setminus X$. Hence, our assumption on $u$ and the above inequality are in conflict. This shows that  no $u^\Pi$-strategyproof SDS can satisfy Condorcet-consistency if $n=3$, $m\geq 4$, and $u(1)-u(2)<u(2)-u(m)$.\medskip
		
		\textbf{Case 2: $u(1)-u(m-1)>u(m-1)-u(m)$}
		
		Next, we assume that $f$ is defined for $n=5$ voters and satisfies $u^\Pi$-strategyproofness for a utility function $u$ with $u(1)-u(m-1)>u(m-1)-u(m)$. In this case, we fix again three alternatives $X=\{x,y,z\}$ and consider the following four preference profiles.\smallskip
		
		\begin{profile}{C{0.07\profilewidth} C{0.29\profilewidth} C{0.29\profilewidth} C{0.29\profilewidth}}
			$R^{1}$: & $1$: $x,*,y,z$ & $2$: $z,*,x,y$ & $3$: $y,*,z,x$ \\
			& $4$: $x,y,z,*$ & $5$: $z,y,x,*$
		\end{profile}
		
		\begin{profile}{C{0.07\profilewidth} C{0.29\profilewidth} C{0.29\profilewidth} C{0.29\profilewidth}}
			$R^{2}$: & $1$: $x,*,z,y$ & $2$: $z,*,x,y$ & $3$: $y,*,z,x$ \\
			& $4$: $x,y,z,*$ & $5$: $z,y,x,*$
		\end{profile}
		
		\begin{profile}{C{0.07\profilewidth} C{0.29\profilewidth} C{0.29\profilewidth} C{0.29\profilewidth}}
			$R^{3}$: & $1$: $x,*,y,z$ & $2$: $z,*,y,x$ & $3$: $y,*,z,x$ \\
			& $4$: $x,y,z,*$ & $5$: $z,y,x,*$
		\end{profile}
		
		\begin{profile}{C{0.07\profilewidth} C{0.29\profilewidth} C{0.29\profilewidth} C{0.29\profilewidth}}
			$R^{4}$: & $1$: $x,*,y,z$ & $2$: $z,*,x,y$ & $3$: $y,*,x,z$ \\
			& $4$: $x,y,z,*$ & $5$: $z,y,x,*$  
		\end{profile}\smallskip

        First, we note that $z$ is the Condorcet winner in $R^2$, $y$ in $R^3$, and $x$ in $R^4$. Hence, Condorcet-consistency entails that $f(R^2,z)=f(R^3,y)=f(R^4,x)=1$. Furthermore, the profile $R^1$ differs from the profiles $R^2$, $R^3$, and $R^4$, in the preference relations of voters $1$, $2$, and $3$, respectively. We thus use $u^\Pi$-strategyproofness to derive conditions on $f(R^1)$. In particular, $u^\Pi$-strategyproofness between $R^1$ and $R^2$ entails the following equation. The left side of this inequality is voter $1$'s expected utility in $R^2$ and the right hand side is his expected utility he could obtain by deviating to $R^1$.
		\begin{align*}
		 u(m\!-\!1)\geq f(R^1,x)u(1)+f(R^1,z)u(m\!-\!1)
		 +f(R^1,y)u(m) +\!\!\sum_{w\in A\setminus X}\!\! f(R^1, w)u(w)\
		\end{align*}
		
		Next, we reformulate the inequality so that our assumption on $u$ can be used. Moreover, we derive symmetric conditions from $R^3$ and $R^4$.
			\begin{align*}
		f(R^1\!,y)(u(m\!-\!1) - u(m))  \geq 
		 f(R^1\!,x)(u(1) - u(m\!-\!1))
		 + \!\!\!\sum_{w\in A\setminus X}\!\! f(R^1\!, w) (u(w) - u(m\!-\!1))\\
		f(R^1\!,x)(u(m\!-\!1) - u(m)) \geq 
		 f(R^1\!,z)(u(1) - u(m\!-\!1))
		 +  \!\!\!\sum_{w\in A\setminus X}\!\! f(R^1\!, w) (u(w) - u(m\!-\!1))\\
		f(R^1\!,z)(u(m\!-\!1) - u(m)) \geq 
	f(R^1\!,y)(u(1) - u(m\!-\!1))
		 + \!\!\! \sum_{w\in A\setminus X}\!\! f(R^1\!, w) (u(w) - u(m\!-\!1))
		\end{align*}
		
		By summing up the last three inequalities, we derive the following equation. 		
		\begin{align*}
		\sum_{w\in X} f(R^1,w) (u(m\!-\!1) - u(m)) \geq 
		& \sum_{w\in X} f(R^1,w)(u(1) - u(m\!-\!1))\\
		&+  3\!\!\!\sum_{w\in A\setminus X} \!\!\!f(R^1, w) (u(w) - u(m\!-\!1))
		\end{align*}
		
		Every alternative $w\in A\setminus X$ is preferred to at least two other alternatives, and thus, $u(w)-u(m-1)>0$. As a consequence, this inequality and our assumption that $u(1)-u(m-1)>u(m-1)-u(m)$ cannot be simultaneously true. Thus, no SDS satisfies both Condorcet-consistency and $u^\Pi$-strategyproofness if $n=5$, $m\geq 4$, and $u(1)-u(m-1)>u(m-1)-u(m)$.\medskip
		
		\textbf{Case 3: Generalizing the impossibility}
		
		Finally, we explain how to extend our base cases a larger number of voters $n$. For odd $n$, this is simple: we can add pairs of voters with inverse preferences to the construction of the required case. These voters do not affect the Condorcet winner as they cancel each other out with respect to the majority margins. Moreover, the remaining analysis only depends on $u^\Pi$-strategyproofness and therefore only on the preferences of specific voters. Hence, no Condorcet-consistent SDS satisfies $u^\Pi$-strategyproofness for any utility function $u\in\mathcal{U}$ if $m\geq 4$, $n\geq 5$, and $n$ is odd. 
		
		For even $n$, we can use Claim (4) of \Cref{prop:properties}. In particular, if we duplicate every voter in the preference profiles used to reason about odd $n$, this claim shows that our analysis stays intact. Moreover, after duplicating, we can again add pairs of dummy voters with inverse preferences without affecting our analysis. Hence, the impossibility also generalizes to even $n\geq 10$. 
	\end{proof}
	
	A close inspection of our proof shows that \Cref{thmc} also holds if $m=3$ unless the considered utility function $u$ is equi-distant, i.e., $u(1)-u(2)=u(2)-u(3)$. This raises the question whether there are $U$-strategyproof and Condorcet-consistent SDSs in this special case. We will next answer this question in the positive: the Condorcet rule ($\cond$), which assigns probability $1$ to the Condorcet winner whenever it exists and returns the uniform lottery over all alternatives otherwise, is $U$-strategyproofness for the set $U=\{u\in\mathcal{U}\colon u(1)-u(2)=u(2)-u(3)\}$ when $m=3$. Moreover, we will show that \cond is effectively the only SDS that satisfies these properties because every Condorcet-consistent and $U$-strategyproof SDS has to agree with \cond for all profiles without majority ties. 
	
	\begin{restatable}{theorem}{thmd}\label{thmd}
    Assume there are $m=3$ alternatives and let $U=\{u\in\mathcal U\colon u(1)-u(2)=u(2)-u(3)\}$. The following claims hold:
    \begin{enumerate}[label=(\arabic*), leftmargin=*]
        \item \cond is $U$-strategyproof. 
        \item If $f$ is $U$-strategyproof and Condorcet-consistent, then $f(R)=\cond(R)$ for all profiles $R$ such that $n_{xy}(R)\neq 0$ for all $x,y\in R$.
    \end{enumerate}
	\end{restatable}
	\begin{proof}
        Assume there are $m=3$ alternatives and let $U=\{u\in \mathcal U\colon u(1)-u(2)={u(2)-u(3)}\}$. We prove both claims of this theorem separately.\medskip

        \textbf{Claim (1):} We first show that \cond satisfies $U$-strategyproofness if $m=3$. Assume for contradiction that this is not true, i.e., that there are preference profiles $R$ and $R'$, a voter $i\in N$, and a utility function $u\in U$ such that $u$ is consistent with $\succ_i$, ${\succ_j}={\succ_j'}$ for all $j\in N\setminus \{i\}$, and $u(\cond(R'))>u(\cond(R))$. We proceed with a case distinction with respect to the existence of a Condorcet winner. First, assume that there is a Condorcet winner~$x$ in $R$, which means that $\cond(R,x)=1$. If another alternative $y$ is the Condorcet winner in $R'$, voter $i$ prefers $x$ to $y$ because he cannot make $y$ into the Condorcet winner otherwise. Consequently, voter $i$ cannot manipulate in this case as $\cond(R',y)=1$ and $u(x)>u(y)$. Next, assume that there is no Condorcet winner in $R'$. Then, we have that $\cond(R',z)=\frac{1}{3}$ for every alternative $z\in A$ and voter~$i$'s expected utility is $u(2)$ since $u(1)=2u(2)-u(3)$. This is only a manipulation if $x$ is voter $i$'s least preferred alternative in $R$, i.e., if $u(x)=u(3)$. However, voter $i$ cannot change that $x$ is the Condorcet winner if he ranks it last, so no manipulation is possible in this case. Finally, assume that there is no Condorcet winner in $R$. Voter $i$'s expected utility in $R$ is again $u(2)$, which means that he can only manipulate by making his best alternative into the Condorcet winner. Because this is not possible, it follows that $\cond$ is $U$-strategyproof for $U=\{u\in \mathcal U\colon u(1)-u(2)=u(2)-u(3)\}$.\medskip

        \textbf{Claim (2):} Next, we show that every other $U$-strategyproof and Condorcet-consistent SDS $f$ agrees with \cond in all profiles for which there are no majority ties. To this end, we observe that Condorcet-consistency immediately implies that $f(R)=\cond(R)$ for all profiles $R$ with a Condorcet winner. Hence, let $R$ denote a profile without a Condorcet winner and without majority ties. This means there is a majority cycle in $R$, i.e., we can label the alternatives such that $n_{xy}(R)>0$, $n_{yz}(R)>0$, and $n_{zx}(R)>0$. Our goal is to show that $f(R,w)=\frac{1}{3}=\cond(R,w)$ for all alternatives $w\in A$.
		
		First, we will analyze the structure of $R$ in more detail. To this end, we enumerate the six possible preference relations and let 
        \begin{align*}
            &{\succ_1}=x,y,z\qquad{\succ_2}=y,z,x\qquad{\succ_3}=z,x,y\\
            &{\succ_4}=z,y,x\qquad{\succ_5}=x,z,y\qquad{\succ_6}=y,x,z.
        \end{align*}
        Moreover, for every $k\in \{1,\dots, 6\}$, we denote by $n_k$ the number of voters in $R$ that submit $\succ_k$. Finally, we define $\delta_1=n_1-n_4$, $\delta_2=n_2-n_5$, and $\delta_3=n_3-n_6$ and aim to show that $\delta_1>0$, $\delta_2>0$, and $\delta_3>0$. For this, we observe that 
        \begin{align*}
		n_{xy}(R)&=n_1-n_4+n_3-n_6 - (n_2-n_5) = \delta_1-\delta_2+\delta_3\\
		n_{yz}(R)&=n_1-n_4+n_2-n_5-(n_3-n_6) = \delta_1+\delta_2-\delta_3\\
		n_{zx}(R)&=n_2-n_5+n_3-n_6-(n_1-n_4)=-\delta_1+\delta_2+\delta_3.
		\end{align*}
        Now, by summing up the first two inequalities, we derive that $n_{xy}(R)+n_{yz}(R)=2\delta_1$. Since we assume that all three majority margins are positive, this means that $\delta_1>0$ and analogous arguments show also that $\delta_2> 0$ and $\delta_3> 0$. We moreover infer from our equations and the assumptions that $n_{xy}(R)>0$, $n_{yz}(R)>0$, and $n_{zx}(R)>0$ that $\delta_1+\delta_2>\delta_3$, $\delta_2+\delta_3>\delta_1$, and $\delta_3+\delta_1>\delta_2$.

		We now that $f(R,x)=f(R,y)=f(R,z)=\frac{1}{3}$ for all profiles $R$ with $n_{xy}(R)>0$, $n_{yz}(R)>0$, and $n_{zx}(R)>0$. Assume for contradiction that there is such a profile $R$ with $f(R)\neq \cond(R)$, which means that $f(R,x)>f(R,y)$, $f(R,y)>f(R,z)$, or $f(R,z)>f(R,x)$. We suppose that $f(R,x)>f(R,y)$ as the remaining cases are symmetric. To derive a contradiction, let $S$ denote a set of voters that report ${\succ_2}=y,z,x$ in $R$ such that $|S|=\delta_2$. Such a set exists as $n_2\geq \delta_2$. Moreover, let $u\in U$ denote a utility function that is consistent with $\succ_2$. Since $u(y)-u(z)=u(z)-u(x)$ and $f(R,x)>f(R,y)$, it holds that  $u(f(R))<u(2)$. 
        Next, let $R'$ denote the preference profile where all voters in $S$ report $z,y,x$ and note that $n_{zx}(R')=n_{zx}(R)>0$ as no voter changed his preference over $x$ and $z$. Moreover, it holds that $n_{yz}(R')=n_{yz}(R)-2\delta_2=\delta_1-\delta_2-\delta_3<0$ because the voters in $S$ prefer $y$ to $z$ in $R$ but revert this preference in $R'$. Hence, $z$ is the Condorcet winner in $R'$ and $f(R',z)=1$ due to Condorcet-consistency. However, this means that the voters in $S$ can manipulate by deviating from $R$ to $R'$, thus contradicting Claim (4) of \Cref{prop:properties}. This is the desired contradiction, thus showing that $f(R)=\cond(R)$ for all profiles without a Condorcet winner and majority ties. 
	\end{proof}
	
	\begin{remark}\label{rem6}
    If $n$ is odd and $m=3$, \Cref{thmd} characterizes the Condorcet rule as the only Condorcet-consistent SDS that is $U$-strategyproof for the set $U=\{u\in\mathcal{U}\colon u(1)-u(2)=u(2)-u(3)\}$. Moreover, no Condorcet-consistent SDS is $u^{\pi}$-strategyproof for a utility function $u$ outside of this set, so \cond is the only Condorcet-consistent SDS that is $U$-strategyproof for a non-empty and symmetric set $U$ if $n$ is odd and $m=3$. By contrast, if $n$ is even, one can define multiple variants of $\cond$ that also satisfy these properties. For instance, the following SDS $f^3$ is $U$-strategyproof for the given set $U$ and Condorcet-consistent: if two alternatives are top-ranked by exactly half of the voters, $f^3$ assign probability $\frac{1}{2}$ to these alternatives. In all other cases, $f^3$ agrees with \cond. 
	\end{remark}

    \begin{remark}
        The Condorcet rule as well as the tradeoff between Condorcet-consistency and strategyproofness has attracted significant attention. For instance, already \citet{Pott70a} suggested the Condorcet rule as a simple and strategyproof rule for $3$ alternatives and \citet{Gard76a} has shown a first impossibility theorem for Condorcet-consistency and strategyproofness in the context of set-valued voting rules. More recently, after the publication of the conference version of this paper, several papers have investigated the compatibility of strategyproofness and Condorcet-consistency for various other strategyproofness notions and proved far-reaching impossibility theorems \citep[e.g.,][]{BLS22c,BLR23a,BrLe24a}. An interesting follow-up question is whether it is possible to unify all of these results with a general theory on when strategyproofness and Condorcet-consistency are (in)compatible. 
    \end{remark}

	\begin{remark}
		A well-known class of SDSs are tournament solutions which only depend on the majority relation ${\succ_M}=\{(x,y)\in A^2 \colon n_{xy}(R)\geq n_{yx}(R)\}$ of the input profile $R$ to compute the outcome. For these SDSs, unanimity and $u^\Pi$-strategyproofness entail Condorcet-consistency. Thus, there are no unanimous and $u^\Pi$-strategyproof tournament solutions, regardless of the utility function $u$, if $m\geq 4$. This is in harsh contrast to results for set-valued social choice, where attractive tournament solutions satisfy various strategyproofness notions \citep[][]{BBH15a}. 
	\end{remark}

    \subsection{\emph{Ex post} Efficiency}\label{subsec:expost}

    As our last contribution, we will analyze the design of $U$-strategyproof and \ep efficient SDSs. In particular, we will show that, when $U$ contains utility functions that are close to indifferent between the first and second best alternatives, then no $U$-strategyproof and \emph{ex post} efficient SDS can be significantly more decisive than the uniform random dictatorship. 
     To formally state this theorem, we will further generalize $k$-unanimity: we say an SDS $f$ is \emph{$(k,\alpha)$-unanimous} for some $k\in \{0,\dots, n\}$ and $\alpha\in [0,1]$ if $f(R,x)\geq \alpha$ for every profile $R$ and alternative $x$ such that at least $n-k$ voters top-rank $x$ in $R$. Less formally, $(k,\alpha)$-unanimity generalizes $k$-unanimity by only guaranteeing a probability of at least $\alpha$ instead of $1$ to alternatives that are top-ranked by at least $n-k$ voters. Thus, $k$-unanimity is equivalent to $(k,1)$-unanimity. We further observe that the uniform random dictatorship is $(k,\frac{n-k}{n})$-unanimous for every $k\in \{0,\dots, n\}$. 
     
     As we show next, if the gap in the utility between $u(1)$ and $u(2)$ is sufficiently small, no $u^\Pi$-strategyproof and \ep efficient SDS can be significantly more decisive than \rd in the sense of $(k,\alpha)$-unanimity for any $k\in \{1,\dots, n-1\}$. We defer the proof of the following theorem to the appendix as it is lengthy.

    \begin{restatable}{theorem}{exposta}\label{thm:expost1}
            For any $k\in \{1,\dots, n-1\}$, $\epsilon>0$, and utility function $u\in\mathcal{U}$ with $u(1)-(2)\leq\frac{\epsilon}{2}(u(2)-u(3))$, there is no \emph{ex post} efficient and $u^\Pi$-strategyproof SDS that satisfies $(k,\frac{n-k}{n}+\epsilon)$-unanimity if $m\geq 3$ and $n\geq 3$.
    \end{restatable}

    When $\frac{u(1)-u(2)}{u(2)-u(3)}$ converges to $0$, \Cref{thm:expost1} shows that no $u^\Pi$-strategyproof and \emph{ex post} efficient SDS can satisfy $(k,\frac{n-k}{n}+\epsilon)$-unanimity for any $k\in \{1,\dots, n-1\}$ and $\epsilon>0$. Put differently, this means that, in the limit of $\frac{u(1)-u(2)}{u(2)-u(3)}$, \rd is a maximally decisive $u^\Pi$-strategyproof SDS that satisfies \emph{ex post} efficiency. Also, we observe that \Cref{thm:expost1} extends \Cref{thmb} to the class of \emph{ex post} efficient SDS at the cost of a more demanding bound on the utility function $u$.

    \begin{remark}
        For $1$-unanimity, we can strengthen \Cref{thm:expost1} by closely analyzing its proof:  no \emph{ex post} efficient SDS satisfies $u^\Pi$-strategyproofness and $1$-unanimity if $u(1)-u(2)\leq \frac{1}{n-2}(u(2)-u(3))$. This result strengthens the main result by \citet{Beno02a} for \emph{ex post} efficient SDS as we show that a single canonical utility function is sufficient for the impossibility. We moreover note that in \citeauthor{Beno02a}'s setting, strategyproofness and $1$-unanimity imply a sufficient amount of efficiency for our proof. That is, we believe that our assumption of \emph{ex post} efficiency is only a minor restriction for \Cref{thm:expost1}. 
    \end{remark}

    \begin{remark}
        A well-known result in randomized social choice, Gibbard's random dictatorship theorem \citep{Gibb77a}, states that the only \sd-strategyproof and \ep efficient SDSs are (non-uniform) random dictatorships. Intuitively, these SDSs pick each voter with a fixed probability and implement the chosen voter's favorite alternative as the winner of the election. In a follow-up work, \citet{Sen11a} has shown that this result still holds when requiring $U$-strategyproofness for a set of utility functions $U$ such that $U$ contains only three types of utility functions: voters either have an arbitrarily large gap between the utilities of their first and second best alternatives, their second and third best alternatives, or their third and fourth best alternatives. Based on the ideas in the proof of \Cref{thm:expost1}, one can show that the random dictatorship theorem even holds when only the first two types of utility functions are permitted. Put differently, this means that the random dictatorship theorems holds for all $m\geq 3$ and $n\geq 3$ even if voters only care about their first alternative or they are effectively indifferent between two favorite alternatives but prefer them severely to all other alternatives.
    \end{remark}

    \begin{remark}
        We leave it open whether \Cref{thm:expost1} is tight. In particular, the negative results in \Cref{subsec:kunanimity,subsec:Condorcet} make it challenging to, e.g., design \ep efficient SDSs that are $1$-unanimous and $u^\Pi$-strategyproof for a utility function $u$ with $u(1)-u(2)>\frac{1}{n}(u(2)-u(3))\}$ as these results show that no such rule exist within the most common classes of SDSs. However, we note that the rule $f^2$ suggested in Remark 3 tightly matches the improved bound for $1$-unanimity given in Remark 8 when $m=3$ and $n=4$, thus demonstrating that our analysis is tight at least for this special case.
    \end{remark}    
	
	\section{Conclusion}
	
	We study social decision schemes (SDSs) with respect to a new strategyproofness notion called $U$-strategyproofness. Whereas the common notion of \sd-strategyproofness is derived by quantifying over all utility functions, $U$-strategyproofness is derived by quantifying only over the utility functions in a specified set $U$. This new strategyproofness notion arises from practical observations as often not all utility functions are plausible and has theoretical advantages because it allows for a more fine-grained analysis than \sd-strategyproofness. Based on $U$-strategyproofness, we analyze the compatibility of strategyproofness,  decisiveness (in the sense of avoiding randomization), and basic fairness (or symmetry) concerns. Specifically, we first present two variants of the uniform random dictatorship called $\rd^k$ and $\omni$, which guarantee to select an alternative with probability $1$ if it is top-ranked by all but $k$ voters and satisfy $U$-strategyproofness if the set $U$ only contains utility functions $u$ for which $u(1)-u(2)$ is sufficiently large. Moreover, we show for rank-based SDSs that the large gap between $u(1)$ and $u(2)$ is required to be strategyproof and has to increase in $k$. Secondly, we prove that $U$-strategyproofness is incompatible with Condorcet-consistency if the set $U$ is symmetric and there are $m\geq 4$ alternatives. Finally, we show that no \ep efficient and $U$-strategyproof SDS can be significantly more decisive than the uniform random dictatorship if the set $U$ contains utility functions that are close to indifferent between the two favorite alternatives. 
	
	Our results have an intuitive interpretation: strategyproofness is only compatible with decisiveness if each voter has a clear best alternative. Even more, the more decisiveness is required, the stronger voters have to favor their most-preferred alternative. This conclusion is highlighted by \Cref{thma,thmb} as well as \Cref{thm:expost1}. Moreover, it coincides with the informal argument that it is easier to manipulate for a voter who deems many alternatives acceptable as he can more easily change the outcome to another acceptable outcome. Hence, our results show that the main source of manipulability are voters who are close to indifferent between some alternatives.

	\section*{Acknowledgments}

    This work was supported by the NSF-CSIRO grant on ``Fair Sequential Collective Decision-Making'' (RG230833). Much of this research was carried out during my PhD at TU Munich, during which I was supported by the Deutsche Forschungsgemeinschaft under grants BR 2312/12-1 and BR 2312/11-2.
	I thank the anonymous IJCAI reviewers and Felix Brandt for helpful comments.

	\appendix
	\clearpage

    \section{Omitted Proofs}
	\subsection{Proof of \Cref{prop:properties}}
        
	\properties*
	\begin{proof}
    Fix a non-empty set of utility function $U$ and suppose that $f$ and $g$ are two $U$-strategyproof SDSs. We will prove each of our four claims independently.\medskip

    \textbf{Proof of Claim (1):} Let $\Pi$ denote the set of all permutations on $A$ and let $U'=\{u\circ\pi\colon u\in U, \pi\in \Pi\}$ be the smallest symmetric set that contains $U$. We assume for this claim that $f$ is neutral and show that it is $U'$-strategyproof. Assume for contradiction that $f$ fails this condition, which means that there are two preference profiles $R$ and $R'$, a voter $i$, a utility function $u'\in U'$ such that ${\succ_j}={\succ_j'}$ for all $j\in N\setminus \{i\}$, $u'$ is consistent with $\succ_i$, and $u'(f(R'))>u'(f(R))$. By the definition of $U'$, there is a permutation $\pi$ and utility function $u\in U$ such that $u'(x)=u(\pi(x))$ for all $x\in A$. 
    Moreover, let $\hat R =\pi(R)$ and $\hat R'=\pi(R')$ denote the preference profiles derived from $R$ and $R'$ by permuting the alternatives with respect to $\pi$, i.e., it holds for all voters $j\in N$ and alternatives $x,y\in A$ that $\pi(x)\mathrel{\hat\succ_j} \pi(y)$ (resp. $\pi(x)\mathrel{\hat\succ_j'} \pi(y)$) if and only if $x\succ_j y$ (resp. $x\succ_j' y$). 
    
    We first note that $u$ is consistent with voter $i$'s preference relation $\hat\succ_i$ in $\hat R$ because $u'$ is consistent with $\succ_i$. In more detail, the consistency of $u'$ and $\succ_i$ shows thus $x\succ_i y$ if and only if $u'(x)>u'(y)$. Moreover, the definitions of $u'$ and $\hat{\succ_i}$ require for all $x,y\in A$ that $u(\pi(x))>u(\pi(y))$ if and only if $u'(x)>u'(y)$, and $\pi(x)\mathrel{\hat\succ_i} \pi(y)$ if and only if $x\succ_i y$. Hence, $u(\pi(x))>u(\pi(y))$ if and only if $\pi(x)\mathrel{\hat\succ_i} \pi(y)$ or equivalently, $u(x)>u(y)$ if and only if $x\mathrel{\hat\succ_i} y$. Next, let  $\pi^{-1}$ be the inverse permutation of $\pi$, i.e., $\pi^{-1}(\pi(x))=x$ for all $x\in A$. Since $\pi^{-1}(\hat R)=R$ and $\pi^{-1}(\hat R')=R'$, it follows from the neutrality of $f$ that $f(\hat R,x)=f(\pi^{-1}(\hat R),\pi^{-1}(x))=f(R,\pi^{-1}(x))$ and $f(\hat R',x)=f(\pi^{-1}(\hat R'),\pi^{-1}(x))=f(R',\pi^{-1}(x))$. Based on this insight and the fact that $u'(x)=u(\pi(x))$, we compute that 
		\begin{align*}
		u(f(\hat{R}))&=\sum_{x\in A} f(\hat{R},x) u(x)
		= \sum_{x\in A} f(R,\pi^{-1}(x)) u(x)= \sum_{x\in A} f(R,x) u(\pi(x))\\
		&<\sum_{x\in A} f(R',x) u(\pi(x))
        =\sum_{x\in A} f(R',\pi^{-1}(x)) u(x)
		=\sum_{x\in A} f(\hat{R'},x) u(x)
        =u(f(\hat{R'})).
		\end{align*}
		
		This means that voter $i$ can manipulate by deviating from $\hat R$ to $\hat R'$ with respect to the utility function $u\in U$, which contradicts our assumption that $f$ is $U$-strategyproof.\medskip

    \textbf{Proof of Claim (2):} For proving this claim, we first observe that if $f$ is $U_1$-strategyproof and $U_2$-strategyproof for two sets of utility functions $U_1$ and $U_2$, then it is $U_1\cup U_2$-strategyproof. This means that it suffices to show for each preference relation ${\succ}\in\mathcal{R}$ that $f$ is $\text{conv}(U\cap \mathcal{U}^{\succ})$-strategyproof. To this end, we fix an arbitrary preference relation $\succ$ and consider two utility functions $u_1,u_2\in U\cap \mathcal{U}^{\succ}$. By $U$-strategyproofness, it follows that $u_1(f(R))\geq u_1(f(R'))$ and $u_2(f(R))\geq u_2(f(R'))$ for all preferences profiles $R$ and $R'$ and voters $i$ such that ${\succ_i}={\succ}$ and ${\succ_j}={\succ_j'}$ for all voters $j\in N\setminus \{i\}$. 
    Now, let $\lambda\in[0,1]$ and define $u_\lambda$ by $u_\lambda(x)=\lambda u_1(x)+(1-\lambda)u_2(x)$ for all $x\in A$. We infer that $u_\lambda(f(R))=\lambda u_1(f(R))+(1-\lambda)u_2(f(R)) \geq \lambda u_1(f(R'))+(1-\lambda)u_2(f(R'))= u_\lambda(f(R'))$ for all profiles $R$ and $R'$ that satisfy the previous requirements. This proves that $f$ is $\{u\}$-strategyproof for every utility function $u\in \text{conv}(U\cap \mathcal{U}^{\succ})$, which equivalently means that it is $\text{conv}(U\cap \mathcal{U}^{\succ})$-strategyproof.\medskip

    \textbf{Proof of Claim (3):} As our third claim, we will show that the set of $U$-strategyproof SDSs is convex. To this end, recall that $f$ and $g$ denote two $U$-strategyproof SDSs and define $h$ by $h(R)=\lambda f(R)+(1-\lambda) g(R)$ for all profiles $R$ and some $\lambda\in[0,1]$. Now, it holds for every preference profile $R$, voter $i$, and utility function $u_i\in U$ that is consistent with $\succ_i$ that $u_i(h(R))=\lambda u_i(f(R))+(1-\lambda) u_i(g(R))$. This immediately implies that $h$ is $U$-strategyproof because $u_i(f(R))\geq u_i(f(R'))$ and $u_i(g(R))\geq u_i(g(R'))$ for all profiles $R$ and $R'$, voters $i$, and utility functions $u_i\in U$ such that $u_i$ is consistent with $\succ_i$ and ${\succ_j}={\succ_j'}$ for all $j\in N\setminus \{i\}$.\medskip

    \textbf{Proof of Claim (4):} We lastly will show that $U$-strategyproofness implies a weak form of group-strategyproofness as groups of voters with the same preference relation cannot benefit by jointly deviating if they have utility functions in $U$. To prove this claim, consider two preference profiles $R$, $R'$, a set of voters $S\subseteq N$, and a utility function $u\in U$ such that ${\succ_j}={\succ_j'}$ for all $j\in N\setminus S$, ${\succ_i}={\succ_j}$ for all $i,j\in S$, and $u$ is consistent with $\succ_i$ for all $i\in S$. Now, consider a sequence of preference profiles $R^0,\dots, R^{|S|}$ such that $R^0=R$, $R^{|S|}=R'$, and $R^{k+1}$ differs from $R^k$ for all $k\in \{0, \dots, |S|-1\}$ by replacing the preference relation $\succ_i$ of a voter $i\in S$ with his preference relation in $R'$. Since $u\in U$ is consistent with the preference relation of every voter $i\in S$, we infer from $U$-strategyproofness that $u(f(R^k))\geq u(f(R^{k+1}))$ for all $k\in \{0,\dots, |S|-1\}$. By chaining these inequalities, it follows that $u(f(R))\geq u(f(R'))$, thus proving our claim. 
	\end{proof}

\subsection{Proof of \Cref{thm:expost1}}

\exposta*
\begin{proof}
Let $f$ denote an \ep efficient and $u^\Pi$-strategyproof SDS for some utility function $u$. We will first discuss some general insights on the behavior of $f$ before proving the theorem. To this end, we define by $\mathcal{D}^{S:x,N\setminus S:y}$ the set of preference profiles where all voters in $S$ top-rank $x$, all voters in $N\setminus S$ top-rank $y$, and $x$ and $y$ are the only Pareto-optimal alternatives in $R$. We will analyze the behavior of $f$ on these profiles in several steps. In particular, we first show that, for all alternatives $x,y\in A$ and groups of voters $S$, it holds that $f(R)=f(R')$ for all $R,R'\in \mathcal{D}^{S:x,N\setminus S:y}$. 
In the second step, we prove a contraction statement: given two sets of voters $S$ and $T$ with $S\neq T$, $S\cap T\neq\emptyset$, and $S\cup T=N$, three alternatives $x,y,z$, and two parameters $\alpha, \beta\in [0,1]$, it holds that $f(R,x)\geq \alpha+\beta-1-\frac{u(1)-u(2)}{u(2)-u(3)}$ for all $R\in \mathcal{D}^{S\cap T:x, N\setminus (S\cap T):z}$ if $f(R^1,x)\geq \alpha$ and $f(R^2,y)\geq\beta$ for all profiles $R^1\in\mathcal{D}^{S:x,N\setminus S:y}$ and $R^2\in\mathcal{D}^{T:y,N\setminus T:z}$. Thirdly, we show a dual expansion lemma: for all sets of voters $S$ and $T$ with $S\cap T=\emptyset$, alternatives $x,y,z\in A$, and parameters $\alpha,\beta\in[0,1]$, it holds that $f(R,z)\geq \alpha+\beta-2\frac{u(1)-u(2)}{u(2)-u(3)}$ for all $R\in\mathcal{D}^{S\cup T:z N\setminus (S\cup T): y}$ if $f(R,x^1)\geq \alpha$ and $f(R^2,z)\geq \beta$ for all $R^1\in\mathcal{D}^{S:x,N\setminus S: y}$ and $R^2\in\mathcal{D}^{T:z,N\setminus T: y}$. By combining these three insights, we will lastly prove the theorem.\medskip

\textbf{Step 1:} Fix a s set of voters $S$ and two alternatives $x,y\in A$. We will show that $f(R)=f(R')$ for all profiles $R,R'\in\mathcal{D}^{S:x,N\setminus S:y}$. To this end, we define $R^*$ as the preference profile given by 
    ${\succ_i}^*=x,y,*$ for all $i\in S$ and ${\succ_i}^*=y,x,*$ for all $i\in N\setminus S$. We will show that $f(R)=f(R^*)$ for all $R\in\mathcal{D}^{S:x,N\setminus S:y}$. This proves the lemma as $R$ is chosen arbitrarily, i.e., it follows from our claim that $f(R')=f(R^*)=f(R)$ for all profiles $R$ and $R'$ in $\mathcal{D}^{S:x,N\setminus S:y}$.
    
	To prove that $f(R)=f(R^*)$, we consider a sequence of preferences profiles $R^0,\dots, R^n$ such that $R^0=R$, $R^n=R^*$, and $R^k$ is derived from $R^{k-1}$ by replacing the preference relation of voter $k$ with $\succ_k^*$ for all $k\in \{1,\dots, n\}$. We note that throughout this sequence, we never change the top-ranked alternative of a voter because every voter top-ranks the same alternative in $R$ and $R^*$. This means that $x\succ_i^k y$ if and only if $x\succ_i y$ for all $i\in N$ and profiles $R^k$. Moreover, it holds for every profile $R^k$ in our sequence that only $x$ and $y$ are Pareto-optimal. Hence, \emph{ex post} efficiency shows that $f(R^k,x)+f(R^k,y)=1$ for all profiles $R^k$. 
    We will next show that $f(R^{k-1})=f(R^k)$ for all $k\in \{1,\dots, n\}$. For this, assume that voter~$k$ top-ranks $x$ in $R$ (and thus also in $R^k$); the argument is symmetric if he top-ranks $y$. Since $f(R^{k-1},x)+f(R^{k-1},y)=f(R^k,x)+f(R^k,y)=1$, $u^{\pi}$-strategyproofness from $R^{k-1}$ to $R^k$ implies that
    \[u(1) f(R^{k-1},x)+u(r(\succ_k,y)) f(R^{k-1},y)\geq u(1) f(R^{k},x)+u(r(\succ_k,y))f(R^{k},y).\]
    Since $1<r(\succ_k,y)$ and $f(R^{k-1},x)+f(R^{k-1},y)=f(R^k,x)+f(R^k,y)=1$, this inequality is only true if $f(R^{k-1},x)\geq f(R^k,x)$. 
    On the other hand, we derive from $u^{\pi}$-strategyproofness from $R^k$ to $R^{k-1}$ that 
 \[u(1) f(R^{k},x)+u(2) f(R^{k},y)\geq u(1) f(R^{k-1},x)+u(2) f(R^{k-1},y).\]
    This inequality is only true if $ f(R^{k},x)\geq f(R^{k-1},x)$. This implies that $f(R^{k-1},x)=f(R^k,x)$ and \emph{ex post} efficiency requires in turn that $f(R^{k-1},y)=f(R^k,y)$ as the probabilities need to sum up to $1$. Hence, $f(R^{k-1})=f(R^k)$ and chaining these equalities proves that $f(R)=f(R^*)$.\medskip

\textbf{Step 2:} For this step, we fix two sets of voters $S$ and $T$ such that $S\neq T$, $S\cap T\neq\emptyset$, and $S\cup T=N$, three distinct alternatives $x,y,z\in A$, and two parameters $\alpha,\beta\in [0,1]$. Moreover, we suppose that $f(R,x)\geq \alpha$ for all profiles $R\in\mathcal{D}^{S:x,N\setminus S:y}$ and $f(R,y)\geq \beta$ for all profiles $R\in\mathcal{D}^{T:y,N\setminus T:z}$, and we will show that $f(R,x)\geq \alpha+\beta-1-\frac{u(1)-u(2)}{u(2)-u(3)}$ for all profiles $R\in\mathcal{D}^{S\cap T:x, N\setminus (S\cap T):z}$. To prove this claim, we first note that, due to Step 1, it suffices to prove that $f(R,x)\geq \alpha+\beta-1-\frac{u(1)-u(2)}{u(2)-u(3)}$ for a single profile $R\in\mathcal{D}^{S\cap T:x, N\setminus (S\cap T):z}$. We hence focus on the following four profiles.\smallskip
	
	\begin{profile}{C{0.07\profilewidth} C{0.3\profilewidth} C{0.3\profilewidth} C{0.3\profilewidth}}
    	$R^1$: & $S\setminus T$: $x,y,z,*$ & $S\cap T$: $x,y,z,*$ & $T\setminus S$: $y,z,x,*$ \\
        $R^2$: & $S\setminus T$: $z,x,y,*$ & $S\cap T$: $y,z,x,*$ & $T\setminus S$: $y,z,x,*$ \\
		$R^3$: & $S\setminus T$: $z,x,y,*$ & $S\cap T$: $x,y,z,*$ & $T\setminus S$: $y,z,x,*$ \\
        $R^4$: &  $S\setminus T$: $z,x,y,*$ & $S\cap T$: $x,y,z,*$ & $T\setminus S$: $z,x,y,*$ 
	\end{profile}

    First, we note that $R^1\in \mathcal{D}^{S:x, N\setminus S: y}$ and $R^2\in \mathcal{D}^{T:y,N\setminus T:z}$, so we have that $f(R^1,x)\geq \alpha$ and $f(R^2,y)\geq \beta$ by assumption. Next, we use $u^{\pi}$-strategyproofness to reason about the outcome for $R^3$. To this end, we observe that, in all four profiles, only $x$, $y$, and $z$ are Pareto-optimal, so \emph{ex post} efficiency requires that $f(R^i,w)=0$ for all alternatives $w\in A\setminus \{x,y,z\}$ and $i\in \{1,2,3,4\}$. We will thus ignore all these alternatives in the subsequent computations. Furthermore, let $p=f(R^3)$ to simplify the notation. Since $R^3$ only differs from $R^1$ in the preference relations of the voters in $S\setminus T$, we infer from $u^\Pi$-strategyproofness (and Claim (4) of \Cref{prop:properties}) that 
    \[p(z) u(1)+ p(x)u(2)+ p(y) u(3)\geq f(R^1,z) u(1)+ f(R^1,x) u(2) + f(R^1,y) u(3).\] 
    Furthermore, it holds that $f(R^1,z) u(1)+ f(R^1,x) u(2) + f(R^1,y) u(3)\geq \alpha u(2)+(1-\alpha) u(3)$ since $f(R^1,x)\geq \alpha$. By chaining our inequalities, we thus get that 
    \[p(z) u(1)+ p(x)u(2)+ p(y) u(3)\geq \alpha u(2)+(1-\alpha) u(3).\] 
    
    Finally, since $\alpha=(p(x)+p(y)+p(z))-(1-\alpha)$, we can rearrange the inequality to 
    \begin{align*}
        p(z)(u(1)-u(2))\geq p(y) (u(2)-u(3))
        - (1-\alpha)(u(2)-u(3)).
    \end{align*}

    We next turn to $u^\Pi$-strategyproofness between $R^2$ and $R^3$. Since these profiles only differ in the preferences of the voters in $S\cap T$, $u^\Pi$-strategyproofness requires that 
    \[p(x) u(1)+ p(y)u(2)+ p(z) u(3)\geq f(R^2,x) u(1)+ f(R^2,y) u(2) + f(R^2,z) u(3).\]

    Because $f(R^2,y)\geq \beta$, we can lower-bound the left side with $\beta u(2)+(1-\beta)u(3)$. By applying analogous transformations as for $R^1$, we hence derive that 
    \begin{align*}
    p(x)(u(1)-u(2))&\geq p(z) (u(2)-u(3)) - (1-\beta)(u(2)-u(3)). 
    \end{align*}

    We will next show that these inequalities require that the voters in $T\setminus S$ have a low expected utility in $R^3$. Specifically, we will derive an upper bound on the expected utility $u(p)=p(y)u(1)+p(z)u(2)+p(a)u(x)$ of these voters. For this, we treat $p(x)$, $p(y)$, and $p(z)$ as variables and aim to maximize $u(p)$ subject to our our previous constraints. Moreover, we include the lottery constraint in our linear program, capturing that the total probability needs to sum up to $1$. Note that we could also add the constraints that $p(w)\geq 0$ for $w\in \{x,y,z\}$ to our linear program, but these constraints will not be necessary for our analysis. \medskip
	
	\begin{tabular}{llr}
		$\max$ & $p(y) u(1) + p(z) u(2) + p(x) u(3)$ \\
		subject to & $p(x)+p(y)+p(z)=1$ & $\quad$(Lottery)\\
		& $p(z)(u(1) - u(2)) \geq p(y)(u(2) - u(3)) - (1-\alpha)(u(2)-u(3))$ & (SP1) \\
		& $p(x)(u(1) - u(2)) \geq p(z)(u(2) - u(3))-(1-\beta)(u(2)-u(3))$& (SP2)\medskip
	\end{tabular}

    We first note that the constraints SP1 and SP2 are tight in an optimal solution of the linear program. If this was not the case, we could move probability from $z$ to $y$ or from $x$ to $z$, which increases our objective value. Hence, we can treat our linear program as an equation system. Specifically, when letting $v=\frac{(1-\alpha)(u(2)-u(3))}{u(1)-u(2)}$, $w=\frac{(1-\beta)(u(2)-u(3))}{u(1)-u(2)}$, and $t=\frac{u(2)-u(3)}{u(1)-u(2)}$, we need to solve the system given by 
    \begin{align*}
        p(x)+p(y)+p(z)=1,\qquad 
        p(z)=tp(y)-v,\qquad p(x)=tp(c)-w.
    \end{align*}
    
    It can be verified that these equations are satisfied if and only if 
    \begin{align*}
    &p(x)=\frac{-v t - w - w t + t^2}{1+t+t^2 },\qquad p(y)=\frac{1+v+w+vt}{1+t+t^2},
        \qquad p(z)=\frac{wt-v+t}{1+t+t^2}.
    \end{align*}

    We hence derive that the expected utility of the voters in $T\setminus S$ is at most
    \begin{align*}
        u(p)&\leq u(1)\frac{1+v+w+vt}{1+t+t^2} + u(2)\frac{wt-v+t}{1+t+t^2} + u(3)\frac{-v t - w - w t + t^2}{1+t+t^2}\\
        &=(u(1)-u(2))\frac{1+v+w+vt}{1+t+t^2} + (u(2)-u(3))\frac{1+w+vt+wt+t}{1+t+t^2}+u(3)\\
        &=(u(1)-u(2))\frac{1+v+w}{1+t+t^2}+(u(1)-u(2))\frac{vt}{1+t+t^2} \\
        &\qquad+ (u(2)-u(3))\frac{1+w}{1+t+t^2} + t(u(2)-u(3))\frac{1+v+w}{1+t+t^2}+u(3)\\
        &=(u(1)-u(2))\frac{1+v+w}{1+t+t^2}+t(u(1)-u(2))\frac{1+v+w}{1+t+t^2} \\
        &\qquad + t^2(u(1)-u(2))\frac{1+v+w}{1+t+t^2}+u(3)\\
        &=(u(1)-u(2))(1+v+w)+u(3).
    \end{align*}

    Here, the first two equations are basic transformations. For the third line, we use that $u(2)-u(3)=t(u(1)-u(2))$ by the definition of $t$ and we cancel the term $1+t+t^2$ in the last line.

    Finally, the voters in $T\setminus S$ can deviate from $R^3$ to $R^4$. By $u^\Pi$-strategyproofness, we thus infer that 
    \[p(y) u(1)+p(z) u(2) + p(x)u(3)\geq f(R^4,y) u(1)+f(R^4,z) u(2) + f(R^4,x)u(3).\]
    
    Based on our upper bound on $u(p)$, it follows that $(u(1)-u(2))(1+v+w)+u(3)\geq f(R^4,y) u(1)+f(R^4,z) u(2) + f(R^4,x)u(3)$. Moreover, it holds that $f(R^4,y)=0$ by \ep efficiency. By combining these insights and subtracting $u(3)$ from both sides, we derive that 
    \begin{align*}
        (u(1)-u(2))(1+u+v)\geq f(R^4,z) (u(2)-u(3)).
    \end{align*}

    Since $f(R^4,x)+f(R^4,z)=1$, this means that $f(R^4,x)\geq 1-\frac{u(1)-u(2)}{u(2)-u(3)}(1+u+v)$. Finally, by substituting the definitions of $u$ and $v$, it follows that 
    \begin{align*}
        f(R^4,x)&\geq 1-\frac{u(1)-u(2)}{u(2)-u(3)}(1+(1-\alpha)\frac{u(2)-u(3)}{u(1)-u(2)}+(1-\beta)\frac{u(2)-u(3)}{u(1)-u(2)})\\
       &=\alpha+\beta-1-\frac{u(1)-u(2)}{u(2)-u(3)}.
    \end{align*}
    This completes the proof of this step.\medskip

    \textbf{Step 3:} Dual to the last step, we will next prove an expansion lemma for two disjoint sets of voters $S$ and $T$. To this end, let $S$ and $T$ denote two groups of voters with $S\cap T=\emptyset$, fix three distinct alternatives $x,y,z\in A$, and let $\alpha,\beta\in [0,1]$. We assume that $f(R,x)\geq \alpha$ for all $R\in\mathcal{D}^{S:x,N\setminus S:y}$ and $f(R,z)\geq \beta$ for all $R\in \mathcal{D}^{T:z, N\setminus T:y}$ and prove that $f(R,z)\geq \alpha+\beta-2\frac{u(1)-u(2)}{u(2)-u(3)}$ for all profiles $R\in\mathcal{D}^{S\cup T:z, N\setminus (S\cup T): y}$. Just as for the last step, it suffices to prove this claim for a single profile $R\in \mathcal{D}^{S\cup T:z, N\setminus (S\cup T): y}$ due to Step 1. Hence, we consider the following four profiles.\smallskip

    \begin{profile}{C{0.08\profilewidth} C{0.3\profilewidth} C{0.3\profilewidth} C{0.3\profilewidth}}
		$R^1$: & $S$: $x,y,z,*$ & $N\setminus (S\cup T)$: $y,z,x,*$ & $T$: $y,z,x,*$\\
		$R^2$: & $S$: $y,x,z,*$ & $N\setminus (S\cup T)$: $y,z,x,*$ & $T$: $z,y,x,*$\\
		$R^3$: & $S$: $x,y,z,*$ & $N\setminus (S\cup T)$: $y,z,x,*$ & $T$: $z,y,x,*$\\
        $R^4$: & $S$: $z,x,y,*$ & $N\setminus (S\cup T)$: $y,z,x,*$ & $T$: $z,y,x,*$	
	\end{profile}

    First, we observe that $R^1\in\mathcal{D}^{S:x,N\setminus S:y}$ and $R^2\in \mathcal{D}^{T:z, N\setminus T:y}$, so $f(R^1,x)\geq \alpha$ and $f(R^2,z)\geq \beta$ by our assumptions. Next, we will analyze the outcome for $R^3$, denoted by $p=f(R^3)$, based on $u^\Pi$-strategyproofness and \ep efficiency. To this end, we first note that all alternatives $w\in A\setminus \{x,y,z\}$ are Pareto-dominated and thus $f(R^3,w)=0$ by \ep efficiency. We will thus ignore these alternatives from now on. Next, we note that the voters in $T$ can deviate from $R^1$ to $R^3$ by reporting $z,y,x,*$ instead of $y,z,x,*$. Hence, $u^\Pi$-strategyproofness requires that 
\begin{align*}
	 f(R^1,y) u(1) +  f(R^1,z)u(2) +  f(R^1,x)u(3) \geq  p(y)u(1) +  p(z)u(2) +  p(x)u(3). 
\end{align*} 

Since $f(R^1,x)\geq \alpha$, it holds that $ f(R^1,y)u(1)+ f(R^1,z)u(2)+ f(R^1,x)u(3)\leq (1-\alpha)u(1)+\alpha u(3)$. Together with our previous inequality, this implies that 
\[(1-\alpha)u(1)+\alpha u(3)\geq  p(y)u(1) +  p(z)u(2) +  p(x)u(3). \]

As the next step, we again use that $\alpha=(p(x)+p(y)+p(z))-(1-\alpha)$. In particular, by substituting $\alpha$ in our last inequality and rearranging the terms, we derive that 
\[(1-\alpha) (u(1)-u(3))\geq (u(1)-u(3))p(y) + (u(2)-u(3))p(z).\]

Lastly, we substitute $p(y)=1-p(x)-p(z)$ and solve for $\alpha$ to derive that 
\[p(x)+\frac{u(1)-u(2)}{u(1)-u(3)}p(z) = p(x)+\frac{u(1)-u(3)-(u(2)-u(3))}{u(1)-u(3)}p(z)\geq \alpha.\]

Next, we turn to the relation between $R^2$ and $R^3$. Since the voters in $S$ can deviate from $R^2$ to $R^3$ by reporting $x,y,z,*$ instead of $y,x,z,*$, we infer from $u^\Pi$-strategyproofness that
\[ f(R^2,y)u(1) +  f(R^2,x)u(2)+ f(R^2,z)u(3)\geq p(y)u(1)+p(x) u(2)+p(z) u(3).\]

Moreover, because $f(R^2,z)\geq \beta$, we can upper bound our right hand side by $f(R^2,y)u(1) +  f(R^2,x)u(2)+ f(R^2,z)u(3)\leq (1-\beta) u(1)+ \beta u(3)$. By applying analogous transformation as for $R^1$, we thus derive that 
\[(1-\beta) (u(1)-u(3))\geq (u(1)-u(3))p(y) + (u(2)-u(3))p(x).\]

Lastly, we substitute again $p(y)$ with $1-p(x)-p(z)$ and solve for $\beta$ to infer that 
\[p(z)+\frac{u(1)-u(2)}{u(1)-u(3)} p(x)=p(z)+\frac{u(1)-u(3)-(u(2)-u(3)}{u(1)-u(3)} p(x)\geq \beta.\]

By summing up our final two inequalities for $R^1$ and $R^3$, it follows that 
\begin{align*}
	& \frac{u(1)-u(3)+u(1)-u(2)}{u(1)-u(3)}\Big(p(x) + p(z)\Big) \geq \alpha + \beta
\end{align*}

Since $\frac{u(1)-u(3)}{u(1)-u(3)+u(1)-u(2)}=1-
\frac{u(1)-u(2)}{u(1)-u(3)+u(1)-u(2)}\geq 1-
\frac{u(1)-u(2)}{u(2)-u(3)}$, we derive that 
\[p(x) + p(z) \geq \alpha+\beta - (\alpha+\beta)\frac{u(1)-u(2)}{u(2)-u(3)}.\]

Now, if $\alpha+\beta\leq 1$, this means that $p(x) + p(z) \geq
\alpha+\beta-\frac{u(1)-u(2)}{u(2)-u(3)}$. 

On the other hand, if $\alpha+\beta>1$, we note that our previous inequalities imply that 
\begin{align*}
    & 1-\alpha \geq  p(y) + \frac{u(2)-u(3)}{u(1)-u(3)} p(z)\geq \frac{u(2)-u(3)}{u(1)-u(3)}(p(y)+p(z))=\frac{u(2)-u(3)}{u(1)-u(3)}(1-p(x))\\
	 & 1-\beta \geq  p(y) + \frac{u(2)-u(3)}{u(1)-u(3)} p(x)\geq \frac{u(2)-u(3)}{u(1)-u(3)}(p(y)+p(x))=\frac{u(2)-u(3)}{u(1)-u(3)}(1-p(z)).
\end{align*}

Using that $\frac{u(1)-u(3)}{u(2)-u(3)}=1+\frac{u(1)-u(2)}{u(2)-u(3)}$, we can rearrange these two inequalities to 
\begin{align*}
    p(x)\geq \alpha-(1-\alpha)\frac{u(1)-u(2)}{u(2)-u(3)} \qquad \text{and }\qquad p(z)\geq \beta-(1-\beta)\frac{u(1)-u(2)}{u(2)-u(3)}.
\end{align*}

By summing up these two inequalities and using that $(1-\alpha)+(1-\beta)\leq 1$ since $\alpha+\beta>1$, it follows again that $p(x)+p(z)\geq \alpha+\beta-\frac{u(1)-u(2)}{u(2)-u(3)}$.

Finally, we will turn to the profile $R^4$. We note for this that $f(R^4,w)=0$ for all $w\in A\setminus \{y,z\}$ due to \ep efficiency. Furthermore $u^\Pi$-strategyproofness from $R^4$ to $R^3$ implies that 
\[f(R^4,z) u(1) + f(R^4,y) u(3)\geq p(z) u(1) + p(x) u(2) + p(y) u(3).\]

By subtracting $u(3)$ from both sides and dividing by $u(1)-u(3)$, this means that 
\[f(R^4,z) \geq p(z)  + p(x) \frac{u(2)-u(3)}{u(1)-u(3)}=p(z)  + p(x)-p(x)\frac{u(1)-(2)}{u(1)-u(3)}.\]

Our lower bound on $p(z)+p(x)$ and the observation that $p(x)\frac{u(1)-u(2)}{u(1)-u(3)}\leq \frac{u(1)-u(2)}{u(2)-u(3)}$ now imply that $f(R^4,z)\geq \alpha+\beta-2\frac{u(1)-u(2)}{u(2)-u(3)}$, thus proving this step.\medskip

\textbf{Step 4:} Finally, we will prove the theorem. To this end, we assume for contradiction that $f$ satisfies \emph{ex post} efficiency, $(k, \frac{n-k}{n}+\epsilon)$-unanimity for some $k\in \{1,\dots, n-1\}$ and $\epsilon>0$, and $u^\Pi$-strategyproofness for some utility function $u$ with $u(1)-u(2)\leq \frac{\epsilon}{2}(u(2)-u(3))$. We proceed with a case distinction regarding the divisibility of $n$. 

First, assume that $k$ divides $n$, which necessitates that $k\leq \frac{n}{2}$. In this case, we first consider two arbitrary sets of voters $S$ and $T$ such that $|S|=|T|=n-k$ and $|S\cap T|=n-2k$ and three arbitrary alternatives $x,y,z\in A$. By $(k,\frac{n-k}{n}+\epsilon)$-unanimity, it holds that $f(R,x)\geq \frac{n-k}{n}+\epsilon$ for all $R\in\mathcal{D}^{S:x,N\setminus S: y}$ and $f(R,y)\geq \frac{n-k}{n}+\epsilon$ for all $R\in\mathcal{D}^{T:y, N\setminus T:z}$. In turn, Step 2 implies that 
\[f(R,x)\geq \frac{2n-2k}{n}+2\epsilon-1-\frac{u(1)-u(2)}{u(2)-u(3)}\geq \frac{n-2k}{n}+\epsilon\]
for all $R\in\mathcal{D}^{S\cap T:x, N\setminus (S\cap T):z}$. For the second inequality, we use here that $\epsilon\geq \frac{\epsilon}{2}\geq\frac{u(1)-u(2)}{u(2)-u(3)}$ by assumption. Next, we note that this inequality holds for all groups $S$ and $T$ and alternatives $x,y,z$, so we have that $f(R,x)\geq \frac{2n-2k}{n}+2\epsilon-1-\frac{u(1)-u(2)}{u(2)-u(3)}\geq\frac{n-2k}{n}+\epsilon$ for all profiles $R\in\mathcal{D}^{S:x, N\setminus S: y}$ with $|S|=n-2k$ and all alternatives $x,y\in A$. Hence, we can now repeat our argument by applying Step 2 to two groups $S$ and $T$ with $|S|=n-2k$, $|T|=n-k$, and $|S\cap T|=n-3k$ and three arbitrary alternatives $x,y,z$ to infer that
\[f(R,x)\geq \frac{2n-3k}{n}+2\epsilon-1-\frac{u(1)-u(2)}{u(2)-u(3)}\geq \frac{n-3k}{n}+\epsilon\]
for all profiles $R\in\mathcal{D}^{S\cap T:x, N\setminus (S\cap T): z}$. More generally, by repeating this argument, it follows that $f(R,x)\geq \frac{n-jk}{n}+\epsilon$ for all profiles $R\in \mathcal{D}^{S:x,N\setminus S: y}$ with $|S|=n-jk$ and arbitrary alternatives $x,y\in A$. Finally, since we assumed that $k$ divides $n$, we derive a contradiction for the profiles $R\in\mathcal{D}^{S:x, N\setminus S:y}$ when $|S|=k$. In more detail, $(k, \frac{n-k}{n}+\epsilon)$-unanimity implies for these profiles that $f(R,y)\geq \frac{n-k}{n}+\epsilon$ and our previous argument shows that $f(R,x)\geq \frac{k}{n}+\epsilon=\frac{n-(n/k-1)k}{n}+\epsilon$. However, this means that $f(R,x) +f(R,y)=1+2\epsilon>1$, contradicting the definition of an SDS. 

As the second case, we suppose that $n-k$ divides $n$, which means that $n-k\leq \frac{n}{2}$. Now, fix two sets of voters $S$ and $T$ such that $|S|=|T|=n-k$ and $S\cap T=\emptyset$, and three distinct alternatives $x,y,z\in A$. By $(k, \frac{n-k}{n}+\epsilon)$-unanimity, it follows that $f(R,x)\geq \frac{n-k}{n}+\epsilon$ for all $R\in\mathcal{D}^{S:x,N\setminus S: y}$ and $f(R,z)\geq \frac{n-k}{n}+\epsilon$ for all $R\in\mathcal{D}^{S:z,N\setminus S: y}$. By Step 3, it hence follows that 
\[f(R,z)\geq \frac{2(n-k)}{n}+2\epsilon-2\frac{u(1)-u(2)}{u(2)-u(3)}\geq \frac{2(n-k)}{n}+\epsilon\]
for all profiles $R\in \mathcal{D}^{S\cup T:z, N\setminus(S\cup T):y}$. The second inequality here follows again as $\frac{\epsilon}{2}\geq \frac{u(1)-u(2)}{u(2)-u(3)}$. Moreover, just in the last case, this means that $f(R,x)\geq \frac{2n-2k}{n}+\epsilon$ for all groups of voters $S$ with $|S|=2(n-k)$, alternatives $x,y\in A$, and profiles $R\in\mathcal{D}^{S:x,N\setminus S: y}$. Hence, we can repeatedly apply Step 3 to infer that $f(R,x)\geq \frac{j(n-k)}{n}+\epsilon$ for all $j\in \{1,\dots, \frac{n}{n-k}-1\}$, groups of voters $S$ with $S=j(n-k)$, alternatives $x,y\in A$, and profiles $R\in\mathcal{D}^{S:x, N\setminus S: y}$. However, this gives a contradiction when $|S|=(\frac{n}{n-k}-1)(n-k)=k$. Specifically, $(k,\frac{n-k}{n}+\epsilon)$-unanimity implies for each profile $R\in \mathcal{D}^{S:x, N\setminus S:y}$ that $f(R,y)\geq \frac{n-k}{n}+\epsilon$ and our previous reasoning shows that $f(R,x)\geq (\frac{n}{n-k}-1)\frac{n-k}{n}+\epsilon=\frac{k}{n}+\epsilon$. However, both inequalities cannot be simultaneously be true as otherwise $f(R,x)+f(R,y)>1$. This is the desired contradiction for this case.

Finally, suppose that neither $n-k$ nor $k$ divides $n$. We will additionally assume that $n-k<\frac{n}{2}$; the case that $n-k>\frac{n}{2}$ follows analogously by starting the argument by using Step 2 instead of Step 3. Now, we will prove this case by repeatedly applying Steps 2 and 3 to derive a contradiction. In more detail, first note that $(k, \frac{n-k}{n}+\epsilon)$-unanimity implies that $f(R,x)\geq \frac{n-k}{n}+\epsilon$ for all groups of voters $S$ with $|S|=n-k$, alternatives $x,y\in A$, and profiles $R\in\mathcal{D}^{S:x, N\setminus S: y}$. Since $n-k<\frac{n}{2}$ by assumption, we can apply Step 3 and our previous reasoning to infer that $f(R,x)\geq \frac{j(n-k)}{n}+\epsilon$ for all $j\in \{1,\dots, \lfloor\frac{n}{n-k}\rfloor\}$, sets of voters $S$ with $|S|=j(n-k)$, alternatives $x,y\in A$, and profiles $R\in\mathcal{D}^{S:x,N\setminus S: y}$. Now, let $k_2=n-(n-k)\lfloor\frac{n}{n-k}\rfloor$ and note that $k_2<n-k$. Our previous argument shows that $f(R,x)\geq \frac{n-k_2}{n}+\epsilon$ for all profiles $\in\mathcal{D}^{S:x,N\setminus S:y}$ with $|S|=n-k_2$. Now, if $k_2$ divides $S$, we can infer a contradiction as shown in the first case. Otherwise, we can repeatedly apply Step 2 to infer that $f(R,x)\geq \frac{n-jk_2}{n}+\epsilon$ for all $j\in \{1,\dots, \lfloor\frac{n}{k_2}\rfloor\}$, groups of voters $S$ with $|S|=n-jk_2$, alternatives $x,y\in A$, and profiles $R\in\mathcal{D}^{S:x,N\setminus S: y}$. 
In particular, it holds for $k_3=k_2\lfloor\frac{n}{k_2}\rfloor$ that $f(R,x)\geq \frac{n-k_3}{n}+\epsilon$ for all sets of voters $S$, alternatives $x,y\in A$, and profiles $R\in\mathcal{D}^{S:x,N\setminus S: y}$. Moreover, we note that $n-k_3<k_2<n-k$, i.e., by repeatedly applying Steps 2 and 3, we always reduce $k_i$ or $n-k_i$. Hence, we will eventually arrive at a value $k_i\in \{1,\dots, n-1\}$ such that $k_i$ or $n-k_i$ divides $n$ and $f(R,x)\geq \frac{n-k_i}{n}+\epsilon$ for all sets of voters $S$ with $|S|=n-k_i$, alternatives $x,y\in A$, and profiles $R\in\mathcal{D}^{S:x, N\setminus S:y}$. At this point, we can apply our base cases to obtain a contradiction, thus proving the theorem.
\end{proof}

\end{document}